\documentclass[12pt,letterpaper, onecolumn, accepted=2025-04-01, email=pyuen103@uottawa.ca]{quantumarticle}

\pdfoutput=1

\usepackage[margin=2.54cm]{geometry}
\usepackage{listings}
\usepackage[dvips,final]{graphicx}
\usepackage{amsmath}
\usepackage{amsthm}
\usepackage{cancel}
\usepackage{amsfonts}
\usepackage{dsfont}
\usepackage{epsfig}
\usepackage{braket}
\usepackage{mathtools}
\usepackage[shortlabels]{enumitem}
\usepackage{amssymb}
\usepackage{epsf}
\usepackage{epsfig}
\usepackage{ragged2e}
\usepackage{mathrsfs}
\usepackage{color}
\usepackage{appendix}
\usepackage{setspace}
\usepackage{comment}
\usepackage{enumitem}
\usepackage{tcolorbox}
\usepackage{graphicx}
\usepackage{algorithm}
\usepackage{algpseudocode}
\graphicspath{{./images/}}
\usepackage[strings]{underscore}
\makeatletter
\renewcommand*\env@matrix[1][*\c@MaxMatrixCols c]{%
  \hskip -\arraycolsep
  \let\@ifnextchar\new@ifnextchar
  \array{#1}}
\makeatother
\usepackage{tikz}
\usetikzlibrary{shapes.geometric, shapes.misc, arrows, decorations.pathreplacing}
\tikzstyle{startstop} = [rectangle, minimum width=0.4cm, minimum height=0.4cm,text centered, draw=black, fill=red!30]
\tikzstyle{io} = [rectangle, minimum width=1cm, minimum height=0.5cm, text centered, draw=black, fill=blue!20]
\tikzstyle{process} = [rectangle, minimum width=1cm, minimum height=0.5cm, text centered, text width=2.2cm, draw=black, fill=orange!30]
\tikzstyle{blank} = [rectangle, minimum width=0.5cm, minimum height=0.5cm, text width=2.4cm, draw=white, fill=white!30]
\tikzstyle{arrow} = [thick,->,>=stealth]

\tikzstyle{inputOutput} = [rectangle, draw=white, fill=white!30]
\tikzstyle{operator} = [rectangle, draw=black, fill=white!30]

\tikzset{meter/.append style={draw, inner sep=5, rectangle, font=\vphantom{A}, minimum width=20, line width=.5,
 path picture={\draw[black] ([shift={(.1,.15)}]path picture bounding box.south west) to[bend left=60] ([shift={(-.1,.15)}]path picture bounding box.south east);\draw[black,-latex] ([shift={(0,.1)}]path picture bounding box.south) -- ([shift={(.18,-.08)}]path picture bounding box.north);}}}
\tikzset{cross/.style={cross out, draw=black, minimum size=2*(#1-\pgflinewidth), inner sep=0pt, outer sep=0pt},
cross/.default={1pt}}
\usepackage{color}
\definecolor{darkblue}{rgb}{0,0,0.5}
\definecolor{darkgreen}{rgb}{0,0.5,0}

\usepackage[
    colorlinks=true,
    linkcolor=darkblue,
    urlcolor=darkblue,
    citecolor=darkgreen
]{hyperref}
\usepackage[noabbrev]{cleveref}

\newcommand{\ccC}{\mathscr{C}}

\DeclareMathOperator{\tr}{Tr}

\crefname{algorithm}{protocol}{protocols}
\Crefname{algorithm}{Protocol}{Protocols}
\crefname{line}{step}{steps}
\Crefname{line}{Step}{Steps}

\newtheorem{theorem}{Theorem}[section]
\newtheorem{corollary}[theorem]{Corollary}

\newtheorem{lem}[theorem]{Lemma}
\newtheorem{prop}[theorem]{Proposition}

\newtheorem{defn}[theorem]{Definition}

\theoremstyle{remark}
\newtheorem{remark}[theorem]{Remark}

\makeatletter
\renewcommand*{\ALG@name}{Protocol}
\makeatother

\makeatletter
\providecommand\theHALG@line{\thealgorithm.\arabic{ALG@line}}
\makeatother



\title{Noise-tolerant public-key quantum money from a classical oracle}
\author{Peter Yuen}
\affiliation{University of Ottawa, Department of Mathematics and Statistics}
\email{pyuen103@uottawa.ca}
\date{}

\begin{document}

\maketitle
\begin{abstract}
    Quantum money is the task of verifying the validity of banknotes while ensuring that they cannot be counterfeited. Public-key quantum money allows anyone to perform verification, while the private-key setting restricts the ability to verify to banks, as in Wiesner's original scheme. The current state of technological progress means that errors are impossible to entirely suppress, hence the requirement for noise-tolerant schemes. We show for the first time how to achieve noise-tolerance in the public-key setting. Our techniques follow Aaronson and Christiano's oracle model, where we use the ideas of quantum error correction to extend their scheme: a valid banknote is now a subspace state possibly affected by noise, and verification is performed by using classical oracles to check for membership in ``larger spaces." Additionally, a banknote in our scheme is minted by preparing conjugate coding states and applying a unitary that permutes the standard basis vectors.
\end{abstract}

\tableofcontents

%~~~~~~~~~~~~~~%
% Introduction %
%~~~~~~~~~~~~~~%
\section{Introduction}

The history of quantum money is essentially as long as quantum information itself. At its core, quantum money aims to address the issue of verifying the validity of banknotes while ensuring that they cannot be counterfeited 
(\emph{i.e.}, \emph{cloned}). Famously, this problem was first addressed by Wiesner in \cite{Wie83} where, put simply, a single banknote consisted of $n$ conjugate coding states, and verification involved measuring each qubit in the basis that it was prepared in (thus, this verification procedure requires the verifier to know how a given banknote was prepared). Without knowledge of how the banknote was prepared, a counterfeiter attempting to produce a second copy of a valid banknote has a success probability of at most $(3/4)^n$, as shown in \cite{MVW13}. Informally, Wiesner's scheme is secure within a model where the verification procedure only says whether the banknote is accepted or rejected (the post-measurement state is not returned). 

Crucially, in Wiesner's scheme, only banks are capable of verifying banknotes (since the ability to verify enables the minting of new, valid banknotes). Quantum money where \emph{anyone} can perform verification of a banknote is known as \emph{public-key quantum money}. Public-key quantum money was first addressed in \cite{Aar09} where it was proved that there exists a quantum oracle relative to which secure public-key quantum money is possible. Then, in the foundational work of \cite{AC12}, it was proven that public-key quantum money relative to a classical oracle is secure. The scheme involves a banknote as a subspace state $\ket{A}$ (a uniform superposition of the vectors in a random $n/2$-dimensional subspace $A \subseteq \mathbb{F}^n_2$) and verification is performed through classical oracles that check for membership in $A$ and the dual subspace $A^\perp$. The technical result from \cite{AC12} is that any quantum algorithm that maps $\ket{A}$ to $\ket{A}^{\otimes 2}$ must make exponentially many queries, implying that a quantum polynomial-time adversary with access to the verification procedure would not be able to counterfeit a banknote.

\paragraph{Noise-tolerant schemes:}
While much work has been devoted to public-key quantum money, that work has been focused on instantiating the oracles of \cite{AC12} or devising entirely new schemes. Regardless of how public-key quantum money is obtained in the plain model (without oracles), it will be necessary for such a scheme to be noise-tolerant (\emph{i.e.}, robust) for it to be practically meaningful. Indeed, one of the most pressing practical issues of quantum money is the necessity of quantum storage for the banknote (for a review of quantum storage, see \cite{HEH+16}; for more recent work, see \cite{WLQ+21} where the authors maintain coherence of a qubit for over one hour). Due to the inherently sensitive nature of quantum information, any amount of time in quantum storage almost guarantees the incursion of quantum errors. If noise affects an honest banknote to the point that it does not pass verification, then it has lost its prescribed value. In the classical oracle scheme of \cite{AC12}, a single error will, in general, cause an honest banknote to no longer pass verification (see \cref{section:construction-of-a-noise-tolerant-mini-scheme}). To the best of our knowledge, we are unaware of any previous work that focuses on noise-tolerant public-key quantum money.

In the private-key setting of quantum money, there has indeed been work towards noise-tolerance (see \cite{PYJ+12, AA17, Kum19}). Notably, the scheme in \cite{AA17} is based off \emph{hidden matching} quantum retrieval games; in this model, the authors prove that, for $n=14$, the maximum achievable noise-tolerance is $23.3\%$, which makes their protocol nearly optimal since it can tolerate up to $23.03\%$ noise. In \cite{BSS21arxiv}, the authors construct a noise-tolerant private tokenized signature scheme (from which one can obtain private-key quantum money; see \cref{section:further-related-work}). To simply describe their approach, a token is a tensor product of conjugate coding states, signing the bit $0$ ($1$) involves measuring all qubits in the computational (Hadamard) basis, and verification amounts to comparing the measured string with a bank's secret string to check for agreement; noise-tolerance comes from allowing inconsistencies in a fraction of the string. Unfortunately, their scheme is not secure against an adversary that can query the classical verification oracle in quantum superposition.

\subsection{Contributions}
We present a scheme for noise-tolerant public-key quantum money and analyze its completeness and soundness error. Intuitively, completeness error quantifies the correctness of the scheme while soundness error quantifies the success probability of a counterfeiter. More precisely, completeness error $\varepsilon$ means that valid banknotes (those minted by the bank and then possibly affected by tolerable noise) are accepted by the verification procedure with probability at least $1-\varepsilon$. Soundness error $\delta$ means that a counterfeiter, who possesses a valid banknote and whose quantum circuit is polynomial in size, can produce another banknote such that both banknotes are accepted with probability at most $\delta$. We show that our scheme has perfect completeness error and has $|\mathcal{E}_q|^2/\textsf{exp}(n)$ soundness error, where $\mathcal{E}_q$ represents the set of tolerated Pauli $X$ and $Z$ errors, each affecting up to $q$ of the $n$ qubits in our banknote (\emph{i.e.}, this corresponds to tolerating $q$ arbitrary errors). Intuitively, perfect completeness error stems from the fact that noise-tolerance is built into our verification procedure (a freshly minted banknote affected by tolerable noise will always be accepted). This verification procedure uses classical oracles to check for membership in spaces that are determined by $q$ and the chosen subspace. As a result, increasing $q$ correlates to increasing the size of the space being checked; intuitively, this improves the success probability of a counterfeiter and is formally reflected in our $|\mathcal{E}_q|^2/\textsf{exp}(n)$ soundness error. A freshly minted banknote in our scheme is a subspace state, as in \cite{AC12}, though we use a relation formalized in \cite{CV22} to show that such a banknote can be prepared by generating conjugate coding states and applying a unitary that permutes the standard basis vectors.

\subsection{Techniques}
\label{section:techniques}
Our approach to noise-tolerance starts with the fact that the banknote state used in \cite{AC12} can be viewed as a codeword for a quantum error correcting code; specifically, a Calderbank-Shor-Steane (CSS) code. With this connection, we ask: is it possible to securely use the corresponding CSS code to detect and tolerate errors? Ignoring security, momentarily, the answer is yes. In the stabilizer formalism \cite{Got97}, error detection involves measuring a set of \emph{stabilizer generators}; the resulting classical string, known as a \emph{syndrome}, can be processed to tell us what errors occurred and on which qubits. This information can be used to correct the noise or, more simply, to accept a state that has been afflicted by correctable noise. When it comes to security, we find that direct application of this error detection procedure, without ``hiding" its properties, can compromise security, as the generators reveal information about the underlying subspace (see \cref{section:direct-error-correction}).

The next and perhaps most obvious approach is to apply a layer of error correction on top of the oracle scheme of \cite{AC12}. With this approach, one would encode the subspace state $\ket{A}$ and then proceed in one of two ways: perform verification in a fault-tolerant manner; or, error correct, decode, and then apply the normal verification procedure. In \cref{section:a-layer-of-error-correction}, we
briefly discuss this approach and outline some of its obstacles, though we do not rule it out as a possibility and instead leave it as one of the open questions listed in \cref{section:future-directions}.

In this paper, the direction we take is to design a scheme that is inherently noise-tolerant (although we focus on tolerating noise, we discuss in \cref{remark:correcting-noise} how our scheme could be easily adapted to also correct the errors). We present two approaches for a noise-tolerant scheme, where one approach can be viewed as a generalization of the other. While the two approaches achieve the same completeness and soundness error, we include both of them here since only one approach seems to have an obvious realization in the plain model with existing techniques, though it is the less general approach (see \cref{section:future-directions} where we discuss this as a direction for future research). Notably, both approaches can be seen as natural extensions of the original classical oracle scheme of \cite{AC12}. The difference between the two approaches, described simply, is that one approach checks for membership in a union of subspaces while the other approach checks membership in a union of subsets. Note that in both approaches, we tolerate up to $q$ arbitrary errors by tolerating up to $q$ Pauli-$X$ errors and up to $q$ Pauli-$Z$ errors, as is done in CSS codes.

Conceptually, the subspace approach closely resembles error detection in the stabilizer formalism. It starts by recognizing that each tolerable syndrome (\emph{i.e.}, tolerable error) corresponds to a subspace. We can then form two unions of disjoint subspaces (one for Pauli-$X$ errors and one for Pauli-$Z$ errors). We then perform verification by using classical oracles to check if the banknote is a member of these two unions.

In the subset approach, we do essentially the same thing except that we check membership in subsets. We recognize that, in the computational basis, a Pauli-$X$ error causes a subspace state to become a superposition of the vectors in a coset (\emph{i.e.}, it becomes a coset state), where each Pauli-$X$ error corresponds to a unique coset. Under this relation, the set of tolerated Pauli-$X$ errors corresponds to a collection of cosets. To then accept a state corrupted by a Pauli-$X$ error, we use a classical oracle to check membership in the subset defined as the union of all cosets from this collection. We tolerate Pauli-$Z$ errors by checking membership in a similarly defined subset. 

We focus primarily on the subset approach, as it is conceptually simpler and the security proof for the subspace approach is almost the same (throughout the security analysis, we remark on any changes that would be necessary for the subspace approach). As one might expect, for a banknote with a fixed size $n$, the more noise we tolerate the easier it will be for a counterfeiter to succeed. This is captured by the $|\mathcal{E}_q|^2/\textsf{exp}(n)$ soundness error, where the quantity $|\mathcal{E}_q|$ (formally given in \cref{eq:number-of-errors}) approaches exponential when $q$ is close to  $n$, and thus the soundness error approaches 1 in this case. On the other hand, if $q=kn$ and $0< k \leq 1/2$, then from \cite{FG06} we see that
\begin{equation}
    |\mathcal{E}_q| \leq 2^{2 H(k) n},
\end{equation}
where $H(\cdot) \in [0,1]$ is the binary Shannon entropy. Thus, if we tolerate only a few errors, then $|\mathcal{E}_q|$ will be sufficiently small to get a ``good" soundness error. This trade-off between security and noise-tolerance is visualized in \cref{fig:soundness-error-trade-off}.

On top of the above construction, we make use of an observation (formalized in \cite{CV22}) on the relation between conjugate coding states and coset states in the generation of our banknotes. Since conjugate coding states are practically simple to prepare, we use this relation with the intent of making banknote generation similarly simple. This relation (\cref{lemma:cosetBB84equivalence}), allows the bank to generate the subspace state banknote by preparing conjugate coding states and then applying a unitary that permutes the standard basis vectors. Intuitively, this does not affect the security of our scheme, as an adversary still only receives a subspace state from the bank.

\subsection{Further background on quantum money}
\label{section:further-related-work}

As noted earlier, Wiesner's scheme is secure in a model where the post-measurement state is not returned. In a model where the verification procedure returns accept/reject \emph{and} the post-measurement state, there is a simple attack on Wiesner's scheme where, by repeatedly submitting banknotes for verification, a counterfeiter can produce a second banknote \cite{Lut10arxiv, Aar09}. More generally, quantum state restoration \cite{FGH+10} can be used to break Wiesner's scheme and schemes like it.

There have been several attempts to obtain public-key quantum money in the plain model, though this has proven to be a highly non-trivial task. Most attempts at instantiating the oracles have subsequently been broken. In the quantum oracle model of \cite{Aar09}, Aaronson proposed an explicit scheme based on random stabilizer states, though this was later broken in \cite{LAF+10}. When the authors of \cite{AC12} introduced the classical oracle model, they also proposed an instantiation that was later broken in \cite{CPFP15, CDF+19, BDG23}. In \cite{Zha19}, Zhandry proposed a new primitive called quantum lightning from which one can obtain public-key quantum money; informally, quantum lightning is like public-key quantum money except that even banknote generation is public, though an adversary is still unable to produce two banknotes with the same serial number. Zhandry also gave a construction of quantum lightning in \cite{Zha19}, though, as noted in a later version of the paper \cite{Zha21}, a weakness in the underlying hardness assumption was detected in \cite{Rob21}. A full attack on the scheme was eventually found in \cite{BDG23}. Also in \cite{Zha19}, Zhandry instantiated the oracles in \cite{AC12} by using indistiguishability obfuscation (iO) \cite{BGI+12, GGH+13, SW14}. The existence of post-quantum iO remains unsolved, as candidate constructions proposed in \cite{WW21, GP21} have been subject to attacks \cite{HJL21}, but results such as \cite{BGMZ18} suggest that it could still exist.

Apart from the previously mentioned schemes that focus on instantiating the oracles, there is a category of proposed public-key quantum money schemes that rely on mostly untested hardness assumptions. A particularly well-known example is a scheme from \cite{FGH+12} that relies on knot theory; while the scheme remains unbroken, it is difficult to analyze and lacks a security proof. Also in this category are the schemes presented in \cite{Kan18arxiv, KSS21arxiv}, though some analysis on \cite{KSS21arxiv} has been done in \cite{BDG23}. 

In \cite{LMZ23, AHY23}, the authors push us towards a better understanding of suitable problems for instantiating public-key quantum money: in \cite{LMZ23}, the authors rule out a natural class of quantum money schemes that rely on lattice problems; and in \cite{AHY23}, the authors rule out a class of money schemes that are constructed from black box access to collision-resistant hash functions (where the verification procedure only makes classical queries to the hash function).

It is also worth mentioning the notion of quantum tokenized signature schemes. The relationship of this primitive to quantum money is characterized in \cite{BDS23} where the authors show that a testable public (private) tokenized signature scheme can give public (private) quantum money. The testable tokenized signature scheme constructed in \cite{BDS23} is based on the oracle scheme in \cite{AC12}. In \cite{CLLZ21}, the authors construct a variant of the \cite{BDS23} scheme (coset states are used instead of subspace states) and show its security in the plain model by using post-quantum iO. The schemes in \cite{BDS23} and \cite{CLLZ21} are not noise-tolerant, and we are unaware of any noise-tolerant public tokenized signature scheme. 

\subsection{Organization}
The paper is organized in the following manner. In \cref{section:preliminaries}, we review relevant material on quantum information, classical oracles, and quantum error correction. In \cref{section:quantum-money}, we review and introduce definitions for quantum money, and we specify our noise-tolerant quantum money scheme. In \cref{section:security}, we present the security analysis of our scheme. In \cref{section:future-directions}, we list a few open questions for future research.

%~~~~~~~~~~~~~~~%
% Preliminaries %
%~~~~~~~~~~~~~~~%
\section{Preliminaries}
\label{section:preliminaries}

We say $f(n) \in \Omega(g(n))$ if there exist constants $c, n_0$ such that
\begin{equation*}
    c g(n) \leq f(n), \quad \forall n > n_0.
\end{equation*}

Given two bit strings $x,y \in \{0,1\}^n$, the notation $x\Vert y$ denotes the concatenation of the two bit strings, \emph{i.e.}, $x\Vert y \coloneqq (x, y) \in \{0,1\}^{2n}$.

For a subspace $A$, we denote its dual (\emph{i.e.}, orthogonal complement) by $A^\perp$.

\subsection{Quantum information}
Let $x \in \{0,1\}^n$ and let $\mathsf{H}$ denote the Hadamard matrix. The $n$-qubit Hadamard transfom is given by
\begin{equation*}
    \mathsf{H}^{\otimes n} \ket{x} = \frac{1}{\sqrt{2^n}}\sum_{z \in \{0,1\}^n} (-1)^{x \cdot z} \ket{z}
\end{equation*}

An arbitrary finite-dimensional Hilbert space is denoted as $\mathcal{H}$. A density operator $\rho$ on $\mathcal{H}$ is a positive semidefinite operator on $\mathcal{H}$ with trace one, $\tr(\rho)=1$.

The single-qubit Pauli-$X$ and -$Z$ operators are, respectively, given by
\begin{equation*}
    X =
    \begin{pmatrix}
    0 & 1\\
    1 & 0
    \end{pmatrix}
    \quad \text{ and } \quad
    Z =
    \begin{pmatrix}
    1 & 0\\
    0 & -1
    \end{pmatrix}.
\end{equation*}

For a linear subspace $A$ of $\mathbb{F}^n_2$ and $t,t' \in \{0,1\}^n$, a subspace state $\ket{A}$ and a coset state $\ket{A_{t,t'}}$ are defined as
\begin{equation}
    \ket{A} = \frac{1}{\sqrt{|A|}}\sum_{v\in A} \ket{v} 
    \quad \text{and} \quad 
    \ket{A_{t,t'}} = X^t Z^{t'} \ket{A} = \frac{1}{\sqrt{|A|}} \sum_{v \in A} (-1)^{v \cdot t'} \ket{v+t},
\end{equation}
where $X^t=X^{t_1} \otimes \cdots \otimes X^{t_n}$, $Z^{t'} = Z^{t'_1} \otimes \cdots \otimes Z^{t'_n}$. The subspace state is a superposition over all the elements in the subspace $A$. In the computational basis, the coset state is a superposition over all elements in the coset $A+t$, while in the Hadamard basis, it is a superposition over all elements in the coset $A^\perp+t'$. 

For $x, \theta \in \{0,1\}^n$, an $n$-tensor product of conjugate coding states is
\begin{equation}
    \ket{x}_{\theta} = \ket{x_1}_{\theta_1} \cdots \ket{x_n}_{\theta_n},
\end{equation}
where $\ket{x_i}_{\theta_i} = \mathsf{H}^{\theta_i}\ket{x_i}$ and $\mathsf{H}$ is the Hadamard gate.

The following lemma formally establishes the relationship between coset states and conjugate coding states. For this lemma, let $\mathcal{B}=\{u_1, \ldots, u_n\}$ be a basis of $\mathbb{F}^n_2$ and let $U_\mathcal{B}$ be a unitary of $(\mathbb{C}^2)^{\otimes n}$ which permutes the standard basis vectors:
\begin{equation}
\label{permuteBasis}
    \forall x \in \{0,1\}^n, \quad \quad U_\mathcal{B}\ket{x} = \Ket{\sum_i x_i u_i}.
\end{equation}

\begin{lem}[\cite{CV22}]
\label{lemma:cosetBB84equivalence}
    Let $\{u_1, \ldots, u_n\}$ be a basis of $\mathbb{F}^n_2$ and $U_\mathcal{B}$ be defined as in \cref{permuteBasis}. Let $\{u^1, \ldots, u^n\}$ be its dual basis, i.e., $u^i\cdot u_j = \delta_{i,j}$. Let $T \subseteq \{1,\ldots,n\}$ be such that $|T|=n/2$ and $A=\text{Span}\{u_i : i\in \overline{T}\}$. Let $\theta \in \{0,1\}^n$ be the indicator of $\overline{T}$ (i.e., $\theta_i=1$ if and only if $i \notin T$), and $x \in \{0,1\}^n$. Let $t=\sum_{i\in T} x_iu_i$ and $t'=\sum_{i \in \overline{T}}x_iu^i$. Then
    \begin{equation}
        \ket{A_{t,t'}} = U_\mathcal{B} \ket{x}_\theta.
    \end{equation}
\end{lem}

The fidelity of quantum states $\rho$ and $\sigma$ is defined as
\begin{equation}
\label{equation:fidelity}
    F(\rho,\sigma) 
    \coloneqq 
    \tr{\sqrt{\rho^{1/2} \sigma \rho^{1/2}}}.
\end{equation}
When we take the fidelity of $\rho$ and a pure state $\ket{\psi}$, \cref{equation:fidelity} reduces to
\begin{equation*}
    F(\rho,\ket{\psi}) = \sqrt{\braket{\psi | \rho | \psi}}.
\end{equation*}
Furthermore, by Uhlmann's theorem,
\begin{equation*}
    F(\rho,\sigma) = \max_{\ket{\psi}, \ket{\phi}} \left|\braket{\psi | \phi}\right|,
\end{equation*}
where the maximization is over all purifications $\ket{\psi}, \ket{\phi}$ of $\rho, \sigma$, respectively. The fidelity of $\rho$ with a subspace $S$ is defined as
\begin{equation}
\label{eq:fidelity-with-subspace}
    F(\rho, S) = \max_{\ket{\psi}, \ket{\phi}} \left|\braket{\psi | \phi}\right|
\end{equation}
where the maximization is over all purifications $\ket{\psi}$ of $\rho$ and all unit vectors $\ket{\phi} \in S$.

\begin{lem}[``Triangle inequality" for fidelity \cite{AC12}]
\label{lemma:triangle-inequality-fidelity}
    Suppose $\braket{\psi | \rho | \psi} \geq 1-\epsilon$ and $\braket{\phi | \sigma | \phi} \geq 1-\epsilon$. Then $F(\rho, \sigma) \leq |\braket{\psi | \phi}| + 2\epsilon^{1/4}$.
\end{lem}

\subsection{Classical oracles}
\label{section:classical-oracles}
Let $f: \{0,1\}^* \rightarrow \{0,1\}$ be some Boolean function. Then a classical oracle $U$ is a unitary transformation of the form
\begin{align*}
    U \ket{x} \bigg(\frac{\ket{0} - \ket{1}}{\sqrt{2}}\bigg) 
    &= 
    \ket{x} \bigg(\frac{\ket{0 \oplus f(x)} - \ket{1 \oplus f(x)}}{\sqrt{2}}\bigg)\\
    \Leftrightarrow
    U \ket{x}\ket{-} 
    &= 
    (-1)^{f(x)}\ket{x}\ket{-}.
\end{align*}
Since the ancillary qubit is left untouched by this operation, we refrain from writing it and simply say $U \ket{x} = (-1)^{f(x)}\ket{x}$. Frequently, we use a classical oracle $U_A$ as a means for checking membership in some subset $A$ in the following sense
\begin{equation*}
    U_A \ket{x} =
    \begin{cases}
        -\ket{x}, \quad \text{if } x \in A\\
        \ket{x}, \quad \text{otherwise}.
    \end{cases}
\end{equation*}
As noted in \cite{AC12}, we can implement $U_A$ as a projector $\mathbb{P}_A$. To do this, we initialize a control qubit $\ket{+}$ to control the application of $U_A$ to $\ket{x}$. After this, we measure the control qubit in the Hadamard basis. If the outcome is $\ket{-}$, then $x \in A$ and the state is accepted, otherwise $x \notin A$ and the state is rejected. Indeed,
\begin{equation*}
    CU_A(\ket{x}, \ket{+}) = \frac{\ket{x}\ket{0} + U_A\ket{x}\ket{1}}{\sqrt{2}}=
    \begin{cases}
        \ket{x}\ket{-} \quad \text{if } x \in A\\
        \ket{x}\ket{+} \quad \text{if } x \notin A.
    \end{cases}
\end{equation*}

Note that, unless otherwise stated, a classical oracle can be queried in a quantum superposition.

%%%%%%%%%%%%%%%%%%%%%%%%%%%%%%%%%%%%%%%%%%%%%%%%%%%%%%

\subsection{Quantum error correction and CSS codes}
\label{section:QECC}

In this section, we use \cite{Got97, NC10, Fuj15, Rof19} to briefly review the stabilizer formalism of quantum error correcting codes. We then follow \cite{NC10} to review a well-known category of quantum error correcting codes known as Calderbank-Shor-Steane (CSS) codes.

\subsubsection{Discretization of quantum errors}
\label{section:discretization-of-quantum-errors}

We start by observing that an arbitrary single-qubit error can be written as a linear combination of the single-qubit Pauli gates,
\begin{equation}
    E = \alpha_I I + \alpha_X X + \alpha_Z Z + \alpha_Y Y.
\end{equation}
By using the fact that $Y=iXZ$, we can rewrite this as
\begin{equation}
\label{eq:discretization-quantum-error}
    E = \alpha_I I + \alpha_X X + \alpha_Z Z + \alpha_{XZ} XZ,
\end{equation}
where the $i$ factor has been absorbed in the $\alpha_{XZ}$ coefficient. Thus, an arbitrary error can be written as a sum from the set $\{I, X, Z, XZ\}$. This is known as the \emph{discretization} (or, \emph{digitization}) of quantum errors. Informally, the quantum error correction procedure involves measurements which collapse the superposition in \cref{eq:discretization-quantum-error} to either an $X$, $Z$, or $XZ$ error; thus, correcting $X$ and $Z$ errors alone is sufficient. This will be made clearer in \cref{section:the-stabilizer-formalism} below, where we discuss the error correction procedure in detail.

Note that an $X$ error is known as a \emph{bit-flip} error, and a $Z$ error is known as a \emph{phase-flip} error.

Before describing arbitrary errors on $n$-qubit states, we introduce some useful notation. Let $\mathcal{E}_X$ ($\mathcal{E}_Z$) denote the set of error vectors that corresponds to bit-flip (phase-flip) errors affecting up to $q$ out of $n$ qubits:
\begin{equation}
\label{eq:sets-of-X-and-Z-errors}
    \mathcal{E}_X = \{e \in \{0,1\}^n : \text{wt}(e) \leq q\}
    \quad \text{ and } \quad
    \mathcal{E}_Z = \{e' \in \{0,1\}^n : \text{wt}(e') \leq q\},
\end{equation}
where $\text{wt}(\cdot)$ is the Hamming weight. So, for example, a bit-flip error affecting up to $q$ out of $n$ qubits is of the form $X^e$ where $e \in \mathcal{E}_X$.
Let $\mathcal{E}_q$ denote the following set, which we refer to as the set of all correctable error vectors,
\begin{equation}
\label{eq:set-of-error-vectors}
    \mathcal{E}_q \coloneqq \{(e, e') : e,e' \in \{0,1\}^n \text{ and } \text{wt}(e) \leq q \text{ and } \text{wt}(e') \leq q\}.
\end{equation}
We can easily see that,
\begin{equation}
\label{eq:number-of-errors}
    |\mathcal{E}_q|
    =
    |\mathcal{E}_{X}| \cdot |\mathcal{E}_Z|
    =
    \left(
    \sum^q_{j=0} \binom{n}{j}
    \right)^2.
\end{equation}
When it is clear from context, we will sometimes drop the $q$ subscript on $\mathcal{E}_q$.

Now when we say that an $n$-qubit state has been affected by $q$ arbitrary errors, we mean that $q$ of the qubits have been affected by an error of the form in \cref{eq:discretization-quantum-error}. Then the operator describing $q$ arbitrary errors on an $n$-qubit state $\ket{\psi}$ is a superposition of Pauli products in the following sense,
\begin{equation}
    \sum_{(e,e') \in A} \alpha_{e,e'} X^e Z^{e'} \ket{\psi},
\end{equation}
where $A$ is some subset of $\mathcal{E}_q$ which contains at least one pair $(e,e')$ that corresponds to a weight $q$ error.

\subsubsection{The stabilizer formalism}
\label{section:the-stabilizer-formalism}
The stabilizer formalism \cite{Got97} can be used to describe a wide and important class of quantum error correcting codes. 

To start, we recall that the $n$-Pauli group $\mathcal{P}_n$ is defined as
\begin{equation}
    \mathcal{P}_n \coloneqq \{\pm 1, \pm i \} \times \{I, X, Y, Z\}^{\otimes n }.
\end{equation}
An $n$-qubit stabilizer $\mathcal{S}$ is an abelian (\emph{i.e.}, commutative) subgroup of $\mathcal{P}_n$,
\begin{equation}
    \mathcal{S} \coloneqq \{S_i \in \mathcal{P}_n \, | -I \notin \mathcal{S} \text{ and } \forall S_i,S_j \in \mathcal{S}, [S_i,S_j]=0 \}.
\end{equation}

Given a stabilizer, the codespace $C_\mathcal{S}$ is the vector space that is stabilized by $\mathcal{S}$ in the sense that it consists of all $n$-qubit states which are fixed by every element of $\mathcal{S}$. Equivalently, $C_\mathcal{S}$ is the common +1 eigenspace of all elements of $\mathcal{S}$,
\begin{equation}
    C_{\mathcal{S}} = \{\ket{\psi} : S_i \ket{\psi} = \ket{\psi}, \,\, \forall S_i \in \mathcal{S}_g\}.
\end{equation}
When the context is clear, we sometimes drop the $\mathcal{S}$ subscript on $C_\mathcal{S}$. Note that the reason we require $-I \notin \mathcal{S}$ and that the elements of $\mathcal{S}$ commute is to ensure that the vector space stabilized by $\mathcal{S}$ is non-trivial.

A stabilizer is often described by its generators $g_i$, which are elements of its maximum independent set $\mathcal{S}_g$. Independence is understood in the sense that no element of $\mathcal{S}_g$ can be expressed as a product of other elements in $\mathcal{S}_g$. The generators may be denoted as $\mathcal{S}_g = \langle g_1, \ldots, g_m \rangle$. Alternatively, the generators can be described by a \emph{check matrix}. This is an $m \times 2n$ matrix where the $i$th row is constructed according to the following rules:
\begin{itemize}
    \item If $g_i$ contains an $I$ on the $j$th qubit, then the $j$th and $n+j$th column elements are 0.
    \item If $g_i$ contains an $X$ on the $j$th qubit, then the $j$th column element is 1 and the $n+j$th column element is 0.
    \item If $g_i$ contains a $Z$ on the $j$th qubit, then the $j$th column element is 0 and the $n+j$th column element is 1.
    \item If $g_i$ contains a $Y$ on the $j$th qubit, then the $j$th and $n+j$th column elements are 1.
\end{itemize}

Since elements of a stabilizer have $\pm 1$ eigenvalues, each generator divides the $2^n$-dimensional Hilbert space into two orthogonal subspaces (the $+1$ eigenspace and the $-1$ eigenspace). Intuitively, we then expect that with $n-k$ generators, the common $+1$ eigenspace (the codespace $C_\mathcal{S}$) has dimension $2^k = 2^{n-(n-k)}$, which is indeed the case.

\begin{prop}
\label{prop:stabilizer-codespace-dimension}
    Suppose a stabilizer $\mathcal{S}$ has $n-k$ generators, so that \newline $\mathcal{S}_g =~\langle g_1, \ldots, g_{n-k} \rangle$. Then the vector space $C_\mathcal{S}$ stabilized by $\mathcal{S}$ is $2^{k}$-dimensional.
\end{prop}

Given a stabilizer $\mathcal{S}$, we may have a set of operators that do not belong to $\mathcal{S}$ and yet commute with every element of $\mathcal{S}$. These operators, known as \emph{logical operators}, can be interpreted as a means of computing on encoded states, as they move elements of the codespace around without moving them ``out" of $C_\mathcal{S}$. Formally, the logical operators are elements of $N(\mathcal{S}) - \mathcal{S}$, where $N(\mathcal{S})$ is the \emph{normalizer},
\begin{equation}
    N(\mathcal{S}) \coloneqq \{L \in \mathcal{P}_n : L S_i L^\dag \in \mathcal{S} \,\,\,\, \forall S_i \in \mathcal{S} \}.
\end{equation}
With $n-k$ generators, there are $k$ pairs of logical operators (or rather, $2k$ logical operators). 

We now describe how to perform quantum error correction using a stabilizer $\mathcal{S}$. Since the common starting point for most quantum computations is the $\ket{0}$ state, it is standard to first prepare the codeword $\ket{0}_L \equiv \ket{0^{\otimes k}}_L \in C_\mathcal{S}$ by projecting $\ket{0} \equiv \ket{0^{\otimes n}}$ onto the $+1$ eigenspace of all generators (see \cite{NC10, Rof19} for details on how to do this). Note that the other codewords in $C_\mathcal{S}$ can be prepared by applying the logical operators to $\ket{0}_L$.

Now suppose that an error $E$ has occurred on our encoded state $\ket{\psi}_L$, leaving us with the state $E\ket{\psi}_L$. The next step is to measure each stabilizer generator $g_i$, $i=1,\ldots,n-k$ through the circuit in \cref{fig:measuring-generator}. Before measurement, this circuit has the effect of mapping the overall state $E\ket{\psi}_L\ket{0}_{A_i}$ in the following way
\begin{equation}
\label{eq:generalStabilizerCode}
    E\ket{\psi}_L\ket{0}_{A_i} \longrightarrow \frac{1}{2}(I^{\otimes n} + g_i)E\ket{\psi}_L\ket{0}_{A_i} + \frac{1}{2}(I^{\otimes n} - g_i)E\ket{\psi}_L\ket{1}_{A_i}.
\end{equation}
\begin{figure}
    \centering
    \begin{tikzpicture}
        \node[inputOutput] at (0,0) (bla1) {\small $E\ket{\psi}_L$};
        \node[operator]    at (2.5,0) (op1) {\normalsize $g_i$};
        \node[inputOutput] at (6,0) (bla2) {};
    
        \node[inputOutput] at (0,-1) (bla3) {\small $\ket{0}_{A_i}$};
        \node[operator]    at (1.5,-1) (op2) {$H$};
        \node[circle, fill=black, inner sep=0pt, minimum size=4pt] at (2.5,-1) {};
        \node[operator]    at (3.5,-1) (op3) {$H$};
        \node[meter]       at (5,-1) (meter)  {};
        \node[inputOutput] at (6,-1) (bla4) {};
        %%%
        \draw[thick] (bla1) -- (op1);
        \draw[thick] (op1) -- (bla2);
    
        \draw[thick] (bla3) -- (op2);
        \draw[thick] (op2) -- (op3);
        \draw[thick] (op3) -- (meter);
        \draw[thick, double] (meter) -- (bla4);
        
        \draw[thick] (op1) -- (2.5,-1);
    \end{tikzpicture}
    \caption{\label{fig:measuring-generator} Circuit to measure a generator $g_i$}
\end{figure}
We say that a stabilizer generator will detect the error $E$ if it anti-commutes with $E$, in which case a measurement of the ancilla qubit yields 1; if they commute, the measurement will yield 0. We can equivalently interpret this as: the generator $g_i$ detects the error $E$ if $E \ket{\psi}_L$ is in the $-1$ eigenspace of $g_i$ (the error has moved the state ``out" of the codespace), and the error will not be detected if $E \ket{\psi}_L$ is in the +1 eigenspace of $g_i$. From this perspective, logical operators can be viewed as undetectable errors.

After each of the $n-k$ ancilla qubits $\ket{0}_{A_i}$ have been measured, the results form an $(n-k)$-bit string called the \emph{syndrome}. At this point the syndrome is processed (or, decoded) to determine an appropriate recovery operation $R$ so that $RE\ket{\psi}_L = \ket{\psi}_L$.

At this stage, we can elaborate on a point we raised earlier: the error correction procedure collapses an arbitrary error to a bit-flip or phase-flip error (or a combination of the two), and so it suffices to correct $X$ and $Z$ errors. For simplicity, suppose we have an error of the form $E=\alpha_I I + \alpha_X X + \alpha_Z Z + \alpha_{XZ} XZ$. Further suppose that $\ket{\psi}_L, X\ket{\psi}_L$ are in the $+1$ eigenspace of $g_i$ while $Z\ket{\psi}_L, XZ \ket{\psi}_L$ are in its $-1$ eigenspace. Then applying the circuit in \cref{fig:measuring-generator}, we get the following overall state immediately before measurement,
\begin{equation}
    (\alpha_I I^{\otimes n} + \alpha_X X) \ket{\psi}_L\ket{0}_{A_i} + (\alpha_Z Z + \alpha_{XZ}XZ) \ket{\psi}_L\ket{1}_{A_i}.
\end{equation}
Then measuring the ancilla qubit has the effect of collapsing the original superposition $E$. By repeating this procedure with more generators, we end up with $\ket{\psi}_L, X\ket{\psi}_L$, $Z\ket{\psi}_L,$ or $XZ\ket{\psi}_L$, and thus the procedure has effectively collapsed the arbitrary error $E$ into an $X$, $Z$ or $XZ$ error, where the syndrome can be thought of as indicating which of these errors we have collapsed to. This also applies to the more general case of $q$ arbitrary errors, where the overall error operator is a superposition of Pauli products of the form $X^e Z^{e'}$ where $(e,e') \in \mathcal{E}_q$.

Informally, the {\em distance} $d$ of a code is the minimum number of errors required to transform one codeword into another. For the formal definition, recall that $N(\mathcal{S})-\mathcal{S}$ is the set of logical operators, and that the weight of an error $\textsf{wt}(E)$ is the number qubits being acted upon by a Pauli operator (excluding the identity). Then the distance is defined as,
\begin{equation}
    d = \min_{E \in N(\mathcal{S})-\mathcal{S}} \textsf{wt}(E).
\end{equation}
The distance of a code plays an important role in quantifying how many errors it can correct. If $d \geq 2q+1$, then the code is capable of correcting $q$ arbitrary errors.
 
Quantum error correcting codes are often described in the format $[[n,k,d]]$, where $n$ is the number of qubits, $k$ is the number logical qubits, and $d$ is the code's distance.

\subsubsection{Calderbank-Shor-Steane (CSS) codes}

We now review a well-known category of quantum error correcting codes known as Calderbank-Shor-Steane (CSS) codes. We start by describing their construction from classical linear codes; after this, we examine CSS codes within the stabilizer formalism.

A classical linear code $C$ that encodes $k$ bits into an $n$-bit codespace is specified by an $n \times k$ matrix $G$ called the \emph{generator matrix}. The entries of $G$ are elements of $\mathbb{Z}_2=\{0,1\}$ and its columns are linearly independent and span the codespace. So, for a $k$-bit string $x$, represented as a column vector, we get the corresponding codeword by computing $Gx$.

Equivalently, a classical code can be defined as the kernel of an $(n-k) \times n$ matrix $H$ whose entries are elements of $\mathbb{Z}_2$. That is, the codespace consists of all $x \in \{0,1\}^n$ such that $Hx=0$. The matrix $H$ is known as a \emph{parity check} matrix.

The \emph{distance} of a code $C$ is defined as
\begin{equation*}
    d(C) \coloneqq \min_{x,y \in C, x \neq y} d(x,y).
\end{equation*}
where $d(x,y)$ is the Hamming distance. Since $d(x,y)=\text{wt}(x+y)$, where $\text{wt}(\cdot)$ is the Hamming weight, we equivalently have
\begin{equation}
\label{eq:distance-of-code}
    d(C) = \min_{x \in C, x\neq 0} \text{wt}(x).
\end{equation}
If a code has distance $d \geq 2q+1$ for some integer $q$, then the code is able to correct up to $q$ errors (for a noisy codeword $y'$, there is a unique codeword $y$ that is within Hamming distance $q$ of $y'$). 

A classical linear code is denoted as $[n, k, d]$. Observe that this resembles the notation for quantum codes, except that the latter case uses double square brackets (to distinguish from the classical case).

Given an $[n,k]$ classical linear code $C$ with generator matrix $G$ and parity check matrix $H$, one can define the dual of $C$, denoted as $C^\perp$, as the code with generator matrix $H^T$ and parity check matrix $G^T$.

Now suppose $C_1$ and $C_2$ are $[n,k_1]$ and $[n,k_2]$ classical linear codes with the property that $C_2 \subseteq C_1$, and $C_1$ and $C^\perp_2$ both correct $q$ errors. The CSS code, constructed from $C_1$ and $C_2$, is then an $[[n,k_1-k_2]]$ quantum code that can correct up to $q$ bit-flip and $q$ phase-flip errors (\emph{i.e.}, up to $q$ arbitrary errors), and thus it corrects the set $\mathcal{E}_q$, defined in \cref{eq:set-of-error-vectors}.

The CSS code is an example of a stabilizer code where its check matrix is of the form
\begin{equation}
\label{equation:check-matrix}
\begin{bmatrix}[c|c]
    H_{C^\perp_2} & 0\\
    0 & H_{C_1}
\end{bmatrix}.
\end{equation}
So, from the parity check matrices $H_{C_1}$ and $H_{C_2^\perp}$ we can obtain the stabilizer generators, and vice versa. The code $C_1$ then contributes $n-k_1$ generators, and the code $C^\perp_2$  contributes $k_2$ generators, for a total of $n-k_1+k_2$ generators. By \cref{prop:stabilizer-codespace-dimension}, the codespace has dimension $2^k=2^{n-(n-k_1+k_2)}=2^{k_1-k_2}$ which is exactly what we saw earlier.

Codewords in the CSS code are coset states of the form
\begin{equation}
\label{CSScodeword}
    \ket{w + C_2} \equiv \frac{1}{\sqrt{|C_2|}}\sum_{v \in C_2} \ket{w+v},
\end{equation}
where $w \in C_1$. Note that since $C_2 \subseteq C_1$, $w+v \in C_1$. The codespace is then the vector space spanned by these coset states. For distinct $w_1, w_2 \in C_1$, the coset states $\ket{w_1 + C_2}, \ket{w_2 + C_2}$ are orthonormal. The number of such cosets is given by $|{C_1}| / |{C_2}| = 2^{k_1 - k_2}$ which is precisely the dimension of the codespace.

Intuitively, error correction for a CSS code works by using the error correcting properties of its classical linear codes; $C_1$ is used to correct bit-flip errors and $C^\perp_2$ is used to correct phase-flip errors. To understand this better, we start by examining a noisy codeword; bit-flip errors are described be an $n$-bit vector $e_1$ with 1s where bit-flips have occurred and 0s otherwise. That is, $e_1 = (e_{1,1}, \ldots, e_{1,n})$ corresponds to the error $X^{e_1} = X^{e_{1,1}} \otimes \cdots \otimes X^{e_{1,n}}$. Similarly, phase-flip errors are described by a vector $e_2$. Then the noisy state is
\begin{equation}
    X^{e_1} Z^{e_2}\ket{w+C_2} = \frac{1}{\sqrt{|C_2|}}\sum_{v \in C_2} (-1)^{(w+v)\cdot e_2} \ket{w+v+e_1}.
\end{equation}

First, bit-flips are dealt with by applying the parity check matrix $H_{C_1}$ for the code $C_1$ and storing the result in ancilla qubits; measuring the ancilla qubits yields an error syndrome $H_{C_1} e_1$ that can be classically processed to determine an appropriate correction operation. After correction, the state is
\begin{equation}
    \frac{1}{\sqrt{|C_2|}}\sum_{v \in C_2} (-1)^{(w+v)\cdot e_2} \ket{w+v},
\end{equation}
By moving to the Hadamard basis, phase-flip errors can be handled in the same way. Indeed, by applying the Hadamard transform, it can be shown that the resulting state takes the form
\begin{equation}
    \frac{1}{\sqrt{|C^\perp_2|}}\sum_{z \in C^\perp_2} (-1)^{w \cdot z}\ket{z+e_2}.
\end{equation}
As in the case of bit-flip errors, we now apply the parity check matrix $H_{C_2^\perp}$ for $C^\perp_2$ and obtain the error syndrome $H_{C_2^\perp} e_2$. After correction, another Hadamard transform is applied to return the state to the original codeword. 

For an $[[n,k]]$ CSS code that can correct up to $q$ arbitrary errors, the relationship between the parameters $n, k,$ and $q$ can be formalized with the Gilbert-Varshamov bound for CSS codes. This bound says that an $[[n,k]]$ CSS code which corrects up to $q$ arbitrary errors exists for some $k$ such that
\begin{equation}
\label{eq:GV-bound}
    \frac{k}{n}
    \geq
    1-2H\left(\frac{2q}{n}\right),
\end{equation}
where $H(x) = -x\log(x)-(1-x)\log(1-x)$ is the binary Shannon entropy.

\section{Noise-tolerance for public-key quantum money}
\label{section:quantum-money}

In this section, we formally define what we mean by a noise-tolerant public-key quantum money scheme and we fully describe our approaches to designing such a scheme. We start in \cref{section:definitions-and-the-standard-construction} with recalling how public-key quantum money was defined in \cite{AC12}. Then, in \cref{section:construction-of-a-noise-tolerant-mini-scheme}, we extend this definition to the noisy setting. Before describing our approaches to designing a noise-tolerant scheme (in \cref{section:formal-specification} and \cref{section:verification}), we briefly outline alternative approaches in \cref{section:initial-approaches-to-noise-tolerance}, though these are not the main focus of the paper.

\subsection{Definitions}
\label{section:definitions-and-the-standard-construction}

In this section, we recall how public-key quantum money is defined in \cite{AC12}. Specifically, we focus on what Aaronson and Christiano call \emph{mini-schemes}, since they show that a secure mini-scheme can be combined with a digital signature scheme to construct a full public-key quantum money scheme that is also secure. It is worth noting that the construction uses a mini-scheme in a straightforward manner, so that if we handle noise at the level of the mini-scheme, we will retain the same level of noise-tolerance in the full scheme (from another perspective, the construction does not involve any transformations of the quantum money state that could result in further errors). For additional details on this construction, the reader is referred to \cite{AC12}.

Notably, a \emph{valid banknote} in the context of \cref{def:mini-schemes} refers to a banknote as output by the algorithm \textsf{Bank}. For the noise-tolerant setting, we define a valid banknote differently (see \cref{defn:valid-banknotes-for-a-noise-tolerant-scheme}).

\begin{defn}[Mini-schemes \cite{AC12}]
\label{def:mini-schemes}
A (public-key) \textbf{mini-scheme} $\mathcal{M}$ consists of two polynomial-time quantum algorithms:
    \begin{itemize}
        \item 
            \textsf{Bank}, which takes as input a security parameter $0^n$, and probabilistically generates a banknote $\$=(s, \rho_s)$ where $s$ is a classical \textbf{serial number}, and $\rho_s$ is a quantum money state.
        \item 
            \textsf{Ver}, which takes as input an alleged banknote $\tilde{\$}$, and either accepts or rejects.
    \end{itemize}

We say $\mathcal{M}$ has \textbf{completeness error} $\varepsilon$ if $\textsf{Ver}(\$)$ accepts with probability at least $1-\varepsilon$ for all valid banknotes \$. If $\varepsilon=0$, then $\mathcal{M}$ has perfect completeness. If, furthermore, $\rho_s = \ket{\psi_s}\bra{\psi_s}$ is always a pure state, and $\textsf{Ver}$ simply consists of a projective measurement onto the rank-1 subspace spanned by $\ket{\psi_s}$, then we say that $\mathcal{M}$ is \textbf{projective}.

Let $\textsf{Ver}_2$ (the \textbf{double verifier}) take as input a single serial number $s$ as well as two (possibly entangled) states $\sigma_1$ and $\sigma_2$, and accept if and only if $\textsf{Ver}(s, \sigma_1)$ and $\textsf{Ver}(s, \sigma_2)$ both accept. We say that $\mathcal{M}$ has \textbf{soundness error} $\delta$ if, given any quantum circuit $\ccC$ of size $\textsf{poly}(n)$ (the \textbf{counterfeiter}), $\textsf{Ver}_2(s, \ccC(\$))$ accepts with probability at most $\delta$. Here, the probability is over the banknote \$ output by $\textsf{Bank}(0^n)$, as well as the behaviour of $\textsf{Ver}_2$ and $\ccC$.

We call $\mathcal{M}$ \textbf{secure} if it has completeness error $\leq 1/3$ and negligible soundness error.
\end{defn}

It is worth noting that the purpose of the classical serial number is to index the quantum money state. This point will be made clearer in \cref{section:formal-specification}.

\subsection{Defining a noise-tolerant mini-scheme}
\label{section:construction-of-a-noise-tolerant-mini-scheme}

Put simply, the banknote in the money scheme of \cite{AC12} is a subspace state $\ket{A}$ and verification is performed by projecting onto the subspace spanned by $\ket{A}$. If the state $\ket{A}$ is corrupted by a Pauli error, then this state will, in general, be orthogonal to $\ket{A}$ (see the end of the proof for \cref{lemma:probability-of-acceptance}) and hence rejected by this verification procedure. To extend the ideas of \cite{AC12} to tolerate noise, we propose that verification of a banknote amounts to projecting onto the subspace spanned by the banknote and all its tolerated noisy variants (\emph{i.e.} a subspace spanned by coset states, where the coset states are indexed by the set of tolerated error vectors). This idea, which is a natural extension of projective mini-schemes from \cref{def:mini-schemes}, is formalized by the following definition.

\begin{defn}[Noisy-projective mini-schemes]
     Let $\mathcal{M}=(\textsf{Bank}, \textsf{Ver})$ be a mini-scheme. We say that $\mathcal{M}$ is \textbf{noisy-projective} if $\textsf{Ver}$ consists of a projective measurement onto the subspace spanned by $\{X^e Z^{e'}\ket{\psi_s}\}_{(e, e') \in \mathcal{E}_q}$, where $\mathcal{E}_q$ is the set of tolerated error vectors for a scheme that can tolerate up to $q$ arbitrary errors on an $n$-qubit banknote (see \cref{eq:set-of-error-vectors}).
\end{defn}

\begin{defn}[Valid banknotes for a noise-tolerant scheme]
\label{defn:valid-banknotes-for-a-noise-tolerant-scheme}
     Let $\mathcal{M}$ be a mini-scheme. In a noisy-projective $\mathcal{M}$, we say that a \textbf{valid banknote} is one that has been minted by the bank and then possibly affected by a tolerable error; that is, a pair consisting of a valid serial number $s$ and a state from the set $\{X^e Z^{e'}\ket{\psi_s}\}_{(e, e') \in \mathcal{E}_q}$.
\end{defn}

\begin{defn}[Security of a noisy-projective mini-scheme]
    Let $\mathcal{M}$ be a mini-scheme. We say that a noisy-projective $\mathcal{M}$ has \textbf{completeness error} $\varepsilon$ if $\textsf{Ver}(\$)$ accepts with probability at least $1-\varepsilon$ for all valid banknotes $\$$ (where valid is understood in the sense of \cref{defn:valid-banknotes-for-a-noise-tolerant-scheme}). If $\varepsilon=0$, then $\mathcal{M}$ has perfect completeness.

    Let $\textsf{Ver}_2$ (the \textbf{double verifier}) take as input a single serial number $s$ as well as two (possibly entangled) states $\sigma_1$ and $\sigma_2$, and accept if and only if $\textsf{Ver}(s, \sigma_1)$ and $\textsf{Ver}(s, \sigma_2)$ both accept. We say that $\mathcal{M}$ has \textbf{soundness error} $\delta$ if, given any quantum circuit $\ccC$ of size $\textsf{poly}(n)$ (the \textbf{counterfeiter}), $\textsf{Ver}_2(s, \ccC(\$))$ accepts with probability at most $\delta$. Here, the probability is over the banknote \$ output by $\textsf{Bank}(0^n)$, as well as the behaviour of $\textsf{Ver}_2$ and $\ccC$.

     When $q\geq 1$, we say that a noisy-projective $\mathcal{M}$ is \textbf{$\epsilon$-noise-tolerant}, where $\epsilon \coloneqq q/n$. We say that a noisy-projective $\mathcal{M}$ is \textbf{$\delta$-secure} if the scheme has soundness error $\delta$ and negligible completeness error.
\end{defn}

\subsection{Initial approaches to noise-tolerance}
\label{section:initial-approaches-to-noise-tolerance}

In this section, we briefly discuss two approaches to noise-tolerance that are at first intuitive but then prove to present obstacles. In \cref{section:direct-error-correction}, we treat the banknote of the \cite{AC12} scheme as already being a codeword of a CSS code. In \cref{section:a-layer-of-error-correction}, we consider a different direction by adding a layer of error correction onto the \cite{AC12} scheme.

\subsubsection{Direct error correction}
\label{section:direct-error-correction}

We can view the banknote subspace state in \cite{AC12} as a codeword in a CSS code. Furthermore, from the perspective of a CSS code (with only one codeword), the verification procedure of \cite{AC12} (checking membership in $A$ and $A^\perp$) can be thought of as error detection where one only checks for membership in the codespace (the common $+1$ eigenspace of the generators), as both procedures amount to projecting onto the subspace spanned by $\ket{A}$. If we use this CSS code and the error correction procedure outlined in \cref{section:QECC} for noise-tolerance, then we must measure the generators of the code to determine whether the codeword has been affected by a correctable error vector. The immediate difficulty with this approach is that publishing the generators reveals the subspace used for constructing the banknote.

To formally see this, we first need to elaborate on why our CSS codes can only have a single codeword. If a CSS code has more than one codeword, then the verification procedure where one checks membership in the codespace allows all codewords to be accepted; in the noise-free quantum money setting, this is advantageous for a counterfeiter, as multiple states can now be accepted, as opposed to the single state $\ket{A}$. In a noise-tolerant setting, this problem would only worsen. 

To be precise in our characterization of applicable CSS codes, we are asking for a subspace $C$ such that $\dim(C)=n/2$ and that both $C$ and $C^\perp$ can correct up to $q$ errors. Then, setting $C_1=C$ and $C_2=C$ in the CSS code framework, we get a code that can correct up to $q$ arbitrary errors and where the codespace has $2^{k_1 - k_2}=1$ element, as desired. The single codeword is then the subspace state
\begin{equation*}
    \ket{C} = \frac{1}{\sqrt{|C|}}\sum_{v \in C} \ket{v}.
\end{equation*}

Returning to the task of noise-tolerance, we recall that measuring the generators of a code is how we determine whether a codeword has been affected by a correctable error vector. For an applicable CSS code, the check matrix (which describes the generators), is
\begin{equation*}
\begin{bmatrix}[c|c]
    H_{C^\perp} & 0\\
    0 & H_{C}
\end{bmatrix}.
\end{equation*}
Thus, if we were to publish a description of the generators for a banknote $\ket{C}$, an adversary could determine the parity check matrices which in turn reveals the codespace $C$. It then seems that to use the generators for tolerating noise, one must ``hide" a full description of them.

\subsubsection{A layer of error correction}
\label{section:a-layer-of-error-correction}

Now suppose that we take a more traditional approach by applying a layer of error correction on top of the oracle scheme of \cite{AC12}. The process for this starts by choosing a stabilizer code and then encoding the banknote subspace state $\ket{A}$. Given a state $\ket{\psi}$, let $\ket{\psi}_L$ denote the encoding of $\ket{\psi}$. The natural way to prepare $\ket{A}_L$ is to construct it from $\ket{0}_L$, just as $\ket{A}$ would be constructed from $\ket{0}$. 

The standard encoding procedure for a stabilizer code is to start by encoding the $\ket{0}$ state via the following operation
\begin{equation}
    \ket{0}_L \equiv \ket{0^{\otimes k}}_L = \frac{1}{N}\prod_{g_i \in \mathcal{S}_g} \left(I^{\otimes n} + g_i\right)\ket{0^{\otimes n}}
\end{equation}
where we recall that $\mathcal{S}_g$ denotes the set of stabilizer generators, and $N$ is for normalization. Note that with the codeword $\ket{0}_L$, one can prepare the other codewords by applying logical operators to $\ket{0}_L$. Now, given the circuit which transforms $\ket{0}$ to $\ket{A}$, one must now translate this circuit into one that operates on the encoded state $\ket{0}_L$ to prepare $\ket{A}_L$.

Assuming that $\ket{A}_L$ can be prepared, one can take two directions. The first is to design a fault-tolerant version of the verification procedure which includes the queries to the classical oracle. The other approach is to perform error correction on $\ket{A}_L$, decode it to $\ket{A}$, and then apply the verification procedure. The latter approach is simpler but assumes that the verification procedure is noise-free. In either approach, a formal proof of security does not seem to immediately follow from the original security proof of \cite{AC12}. A main difficulty is that a codeword may be quite different from the state it encodes, and so when a layer of error correction is applied to the \cite{AC12} scheme, the adversary would hold the encoded state $\ket{A}_L$, while the original security proof of \cite{AC12} assumes that the adversary holds the state $\ket{A}$. Thus, to prove security for this approach, it must be shown that $\ket{A}_L$ provides no additional advantage over holding $\ket{A}$. More formally, the proof of security uses the inner product adversary method (\cref{section:inner-product-adversary-method}) where the number of queries needed by the adversary crucially relies on the state that the adversary starts with, and thus giving $\ket{A}_L$ as the starting point instead of $\ket{A}$ could reduce the number of queries needed.

A natural question is if applying a layer of error correction would be preferable to the approach taken in this paper. Applying a layer of error correction often comes with the practical difficulty of requiring many additional qubits, though it could also come with a favourable trade-off (between noise-tolerance and soundness error). To make a formal comparison, we would need a proof of security for a scheme that accounts for a layer of error correction. We return to this question as a future direction in \cref{section:future-directions}.

%%%%%%%%%%%%%%%%%%%%%%%%%%%%%%%%%%%%%%%%%%%%%%%%%%%%%%

\subsection{Formal specification of the scheme}
\label{section:formal-specification}

In this section, we give the formal specification of our noisy-projective mini-scheme. The specification given here is made in reference to the subset approach, which relies on classical oracles, denoted as $U_{C_{\mathcal{E}_X}}$ and $U_{C_{\mathcal{E}_Z}}$, to perform verification. The formal specification for the subspace approach is similar except that $U_{Q_{\mathfrak{S}^X}}$ and $U_{Q_{\mathfrak{S}^Z}}$ are used, respectively, in place of $U_{C_{\mathcal{E}_X}}$ and $U_{C_{\mathcal{E}_Z}}$. The behaviour of these oracles and how they actually verify a banknote will be made clear when we describe the two approaches in \cref{section:verification}.

To start, we must clarify how the bank actually provides access to the oracles $U_{C_{\mathcal{E}_X}}$ and $U_{C_{\mathcal{E}_Z}}$ in the subset approach. In the final noisy-projective mini-scheme $\mathcal{M} = (\textsf{Bank}_\mathcal{M}, \textsf{Ver}_\mathcal{M})$, the bank, verifier, and counterfeiter will all have access to a single classical oracle $U$, which consists of the four components described below. In \cref{section:using-conjugate-coding-states}, we show how this specification can be amended for the purpose of using conjugate coding states. 

\begin{itemize}
    \item 
        A \textbf{banknote generator} $\mathcal{G}(r)$ which takes as input a random string $r \in \{0,1\}^n$, and outputs a set of linearly independent generators $\langle C_r \rangle = \{x_1, \ldots, x_{n/2}\}$ for a subspace $C_r \in \mathcal{W}$, as well as a unique $3n$-bit serial number $z_r \in \{0,1\}^{3n}$. The function $\mathcal{G}$ is chosen uniformly at random, subject to the constraint that the serial numbers are all distinct. \footnote{As noted in \cite{AC12}, the banknote generator $\mathcal{G}$ could be implemented with a random oracle. In this implementation, the random oracle is queried with the input $r$ and outputs a fixed but random pair $(\langle C_r \rangle, z_r)$. The constraint that the serial numbers are all distinct would be satisfied with probability $1-\mathcal{O}(2^{-n})$.}
    \item 
        A \textbf{serial number checker} $\mathcal{Z}(z)$, which outputs 1 if $z=z_r$ is a valid serial number for some $\langle C_r \rangle$, and 0 otherwise.
    \item 
        A \textbf{primal subset tester} $\mathcal{T}_{\text{primal}}$ which takes an input of the form $\ket{z}\ket{x}$, applies $U_{C_{\mathcal{E}_X}}$ to $\ket{x}$ if $z=z_r$ is a valid serial number for some $\langle C_r \rangle$, and does nothing otherwise.
    \item 
        A \textbf{dual subset tester} $\mathcal{T}_{\text{dual}}$, identical to $\mathcal{T}_{\text{primal}}$ except that it applies $U_{C_{\mathcal{E}_Z}}$ instead of $U_{C_{\mathcal{E}_X}}$.
\end{itemize}
Then $\mathcal{M} = (\textsf{Bank}_\mathcal{M}, \textsf{Ver}_\mathcal{M})$ is defined as follows:
\begin{itemize}
    \item 
        $\textsf{Bank}_\mathcal{M}(0^n)$ chooses $r \in \{0,1\}^n$ uniformly at random. Following this, it looks up $\mathcal{G}(r)=(z_r, \langle C_r \rangle)$, and outputs the banknote $\ket{\$_r} = \ket{z_r} \ket{C_r}$.
    \item 
        $\textsf{Ver}_\mathcal{M}(\tilde{\$})$ first uses $\mathcal{Z}$ to check that $\tilde{\$}$ has the form $(z, \rho)$, where $z=z_r$ is a valid serial number. If so, then it uses $\mathcal{T}_{\text{primal}}$ and $\mathcal{T}_{\text{dual}}$ to apply $V_{C_r} = \mathsf{H}^{\otimes n} \mathbb{P}_{C_r^\perp} \mathsf{H}^{\otimes n} \mathbb{P}_{C_r}$, and accepts if and only if $V_{C_r}(\rho)$ accepts.
\end{itemize}

Recall that the purpose of the serial number is to index the money state. In the above formal specification, this is made clearer, as we see how $z_r$ is used to determine which classical oracle to use for verification (for example, see the primal subset tester $\mathcal{T}_\text{primal}$).

\subsubsection{Using conjugate coding states}
\label{section:using-conjugate-coding-states}

Given that conjugate coding states are practically simple to prepare, it would be ideal if the money state of our scheme could be prepared with conjugate coding states, rather than directly preparing a subspace state $\ket{A}$. In this section, we show how this can be done with \cref{lemma:cosetBB84equivalence} and, consequently, alter the banknote generator $\mathcal{G}(r)$ and $\textsf{Bank}_{\mathcal{M}}$.

Recalling \cref{lemma:cosetBB84equivalence}, we see that generating a uniformly random coset state $\ket{A_{t,t'}}$ over subspaces $A \subseteq \mathbb{F}^n_2$ with dimension $n/2$ is equivalent to generating a basis $\mathcal{B}=\{u_1, \ldots, u_n\}$ and $x,\theta \in \{0,1\}^n$, such that $\theta$ has Hamming weight $n/2$, uniformly at random. This leads to the following procedure:
\begin{enumerate}
    \item 
        Choose $n\geq 2$ to be even. Select $x \in \{0,1\}^n$ uniformly at random and select $\theta \in\{0,1\}^n$, such that $\theta$ has Hamming weight $n/2$, uniformly at random. Then form the state $\ket{x}_\theta$.
    \item 
        Select a uniformly random basis $\mathcal{B}=\{u_1, \ldots, u_n\}$ of $\mathbb{F}^n_2$.
    \item 
        Set $T = \{i: \theta_i = 0\}$ and then set $A=\text{Span}\{u_i: \theta_i=1\}=\text{Span}\{u_i: i \in \overline{T}\}$. Then $A$ is a uniformly random subspace of $\mathbb{F}^n_2$ of dimension $n/2$. \label{generateSubspace}
    \item 
        Set $t=\sum_{i\in T} x_iu_i$ and $t'=\sum_{i\in \overline{T}}x_iu^i$. \label{generateParameters}
    \item 
        By \Cref{lemma:cosetBB84equivalence}, applying $U_{\mathcal{B}}$ to $\ket{x}_\theta$ yields $\ket{A_{t,t'}}$.
\end{enumerate}

In our scheme, a freshly minted money state (not yet corrupted by noise) is a subspace state $\ket{C}$ with $C \in \mathcal{W}$. To generate a uniformly random subspace state $\ket{C}$, we need $t=t'=0$ in the above construction. A non-trivial choice of $x$ cannot achieve this without contradicting the fact that $\{u_1, \ldots, u_n\}$ forms a basis, thus we must have $x=0$. So, to generate a uniformly random subspace state under the above construction, we must generate $\ket{0}_\theta$. 

To then incorporate conjugate coding states into the formal specification of our scheme, we must alter the behaviour of the banknote generator $\mathcal{G}(r)$ and $\textsf{Bank}_\mathcal{M}$.

\begin{itemize}
    \item 
        The \textbf{banknote generator} $\mathcal{G}(r)$ takes as input a random string $r \in \{0,1\}^n$, and outputs a basis $\mathcal{B}_r = \{u_1, \ldots, u_n\}$ and a $\theta_r \in \{0,1\}^n$ with Hamming weight $n/2$ such that the subspace specified by $\theta_r$ belongs to $\mathcal{W}$, as well as a unique $3n$-bit serial number $z_r \in \{0,1\}^{3n}$. The function $\mathcal{G}$ is chosen uniformly at random, subject to the constraint that the serial numbers are all distinct.
    \item 
        $\textsf{Bank}_\mathcal{M}(0^n)$ chooses $r \in \{0,1\}^n$ uniformly at random. Following this, it looks up $\mathcal{G}(r)=(z_r, \theta_r, \mathcal{B}_r)$, and outputs the banknote $\ket{\$_r} = \ket{z_r} U_{\mathcal{B}_r}\ket{0}_{\theta_r}$.
\end{itemize}
Intuitively, using conjugate coding states in the minting of a banknote does not affect the security of the scheme since a counterfeiter in both variants of the scheme only receives a subspace state and access to the oracles. Formally, this change is handled in \cref{theorem:security-mini-scheme}.

\subsection{Verification}
\label{section:verification}
In this section, we describe our noise-tolerant quantum money scheme. We present two approaches which we refer to as the the subset approach (\cref{section:the-subset-approach}) and the subspace approach (\cref{section:the-subspace-approach}). In each section, we characterize the banknote and the verification procedure. As noted in \cref{section:techniques}, we cover both approaches here because the subset approach is more general while the subspace approach seems to have an obvious instantiation (Zhandry's technique of using iO to instantiate the oracle from \cite{AC12} could be used); we discuss this latter point further in \cref{section:future-directions} as a future direction. Since the subset approach is more general, and conceptually simpler, it will be the primary focus for the rest of the paper. In \cref{section:security}, we analyze the security of our scheme.

\subsubsection{The subset approach}
\label{section:the-subset-approach}

We now describe the verification procedure that we use throughout the rest of the paper. To start, let $\mathcal{K}$ be the set of all subspaces $C \subseteq \mathbb{F}^n_2$ such that $\dim C = n/2$. Recalling \cref{eq:distance-of-code}, let $\mathcal{W}$ be the subset of $\mathcal{K}$ defined as
\begin{equation}
\label{eq:applicable-CSS-codes}
    \mathcal{W} \coloneqq \{C \in \mathcal{K} : d(C) \geq 2q+1 \text{ and } d(C^\perp) \geq 2q+1\},
\end{equation}
In other words, $\mathcal{W}$ consists of applicable CSS codes: $n/2$-dimensional subspaces $C$ of $\mathbb{F}^n_2$ such that, as classical linear codes, both $C$ and $C^\perp$ can each correct $q$ errors. 

A freshly minted money state of our scheme will be a subspace state
\begin{equation}
\label{eq:freshly-minted-money-state}
    \ket{C} = \frac{1}{\sqrt{|C|}}\sum_{v \in C} \ket{v},
\end{equation}
where $C \in \mathcal{W}$. This is similar to \cite{AC12}, where a freshly minted money state is also a subspace state except with $C \in \mathcal{K}$. The reason we must select our subspaces from $\mathcal{W}$ is to ensure that our verification operator, constructed below, is a projector; this is made clear in the proof of \cref{lemma:probability-of-acceptance}. We discuss the existence of such subspaces in the subspace approach (\cref{section:the-subspace-approach}), which also requires the $C \in \mathcal{W}$ condition.

When the state in \cref{eq:freshly-minted-money-state} is corrupted by noise (a bit-flip error $X^e$ and a phase-flip error $Z^{e'}$), it becomes a coset state
\begin{equation*}
    \ket{C_{e,e'}} = X^e Z^{e'} \ket{C}
    =
    \frac{1}{\sqrt{|C|}}\sum_{v \in C} (-1)^{v\cdot e'} \ket{v+e},
\end{equation*}
where $e,e'$ are $n$-bit vectors. Suppose we wish to tolerate up to $q$ arbitrary errors. We know that, for quantum error correction (and CSS codes, in particular), we can achieve this by tolerating up to $q$ bit-flip and $q$ phase-flip errors (\emph{i.e.}, tolerating the set $\mathcal{E}_q$). Taking inspiration from this, we design our scheme to tolerate the set $\mathcal{E}_q$ and then show that, as a result, our scheme indeed tolerates $q$ arbitrary errors (though the reasoning for this is not exactly the same as in quantum error correction). We discuss this further after \cref{lemma:probability-of-acceptance}.

When we say that our scheme should tolerate the set $\mathcal{E}_q$, we mean that we should tolerate the sets $\mathcal{E}_X, \mathcal{E}_Z$ defined by \cref{eq:sets-of-X-and-Z-errors}. To tolerate these two sets, we use classical oracles $U_{C_{\mathcal{E}_X}}$ and $U_{C_{\mathcal{E}_Z}}$ to check for membership in the subsets $C_{\mathcal{E}_X}$ and $C_{\mathcal{E}_Z}$, respectively, where these subsets are defined as
\begin{equation}
\label{eq:subset-collections}
    C_{\mathcal{E}_X} \coloneqq \bigcup_{e \in \mathcal{E}_X} C+e 
    \quad \text{ and } \quad
    C_{\mathcal{E}_Z} \coloneqq \bigcup_{e' \in \mathcal{E}_Z} C^\perp+e'.
\end{equation}

Formally, tolerable bit-flip errors are checked by our first projector $\mathbb{P}_{C}$ which is constructed from the classical oracle $U_{C_{\mathcal{E}_X}}$. This construction uses the same techniques of \cite{AC12} for implementing a classical oracle as a projector (recall \cref{section:classical-oracles} for details). In this context, the construction of $\mathbb{P}_C$ consists of the following steps:
\begin{enumerate}
    \item 
        For the oracle $U_{C_{\mathcal{E}_X}}$, we introduce a control qubit $\ket{+}$ and apply the classical oracle controlled on $\ket{+}$,
        \begin{equation}
            CU_{C_{\mathcal{E}_X}}(\ket{x}, \ket{+}) = \frac{\ket{x}\ket{0} + U_{C_{\mathcal{E}_X}}\ket{x}\ket{1}}{\sqrt{2}}=
            \begin{cases}
                \ket{x}\ket{-} \quad \text{if } x \in {C_{\mathcal{E}_X}}\\
                \ket{x}\ket{+} \quad \text{if } x \notin {C_{\mathcal{E}_X}}.
            \end{cases}
        \end{equation}
    \item 
        We then measure the control qubit and accept if the outcome $\ket{-}$ is observed, and reject otherwise.
\end{enumerate}

The projector $\mathbb{P}_{C^\perp}$, which will check for tolerable phase-flip errors, is constructed in exactly the same way except that the classical oracle $U_{C_{\mathcal{E}_Z}}$ is used and, for this projector to work, it will need to be preceded and followed by the Hadamard transform.

Verification is then performed by the following operator
\begin{equation}
\label{eq:verification-operator-subset-approach}
    V = \mathsf{H}^{\otimes n} \mathbb{P}_{C^\perp} \mathsf{H}^{\otimes n} \mathbb{P}_{C}.
\end{equation}

The following lemma is a generalization of lemma 19 from \cite{AC12}. It shows that our verification operator is a projector onto the subspace spanned by the money state $\ket{C}$ and all its tolerated noisy variants. By linearity, the proof of this result extends to arbitrary errors (see the discussion that follows the proof), and so perfect completeness follows from this lemma and the specification given in \cref{section:formal-specification}.

\begin{lem}
\label{lemma:probability-of-acceptance}
    The verification operator $V$, defined in \cref{eq:verification-operator-subset-approach}, is the projector \newline\mbox{$V=\sum_{(e, e') \in \mathcal{E}}\ket{C_{e,e'}}\bra{C_{e,e'}}$}, i.e., it projects onto the subspace $S$ spanned by $\{\ket{C_{e,e'}}\}_{(e,e') \in \mathcal{E}}$. In particular, the probability of acceptance is given by
    \begin{equation*}
        \Pr[V(\ket{\psi}) = \text{accepts}] = \sum_{(e, e') \in \mathcal{E}} |\braket{\psi | C_{e,e'}}|^2.
    \end{equation*}
\end{lem}

\begin{proof}
We first show that $V\ket{C_{e,e'}} = \ket{C_{e,e'}}$ for any $(e,e') \in \mathcal{E}$.
    \begin{align*}
        V\ket{C_{e,e'}} &= V \frac{1}{\sqrt{|C|}}\sum_{v \in C} (-1)^{v\cdot e'} \ket{v+e}\\
        &= \mathsf{H}^{\otimes n} \mathbb{P}_{C^\perp} \mathsf{H}^{\otimes n} \mathbb{P}_{C} \frac{1}{\sqrt{|C|}}\sum_{v \in C} (-1)^{v\cdot e'} \ket{v+e}\\
        &= \mathsf{H}^{\otimes n} \mathbb{P}_{C^\perp} \mathsf{H}^{\otimes n} \frac{1}{\sqrt{|C|}}\sum_{v \in C} (-1)^{v\cdot e'} \ket{v+e}\\
        &= \mathsf{H}^{\otimes n} \mathbb{P}_{C^\perp} \frac{1}{\sqrt{|C|}}\sum_{v \in C} (-1)^{v\cdot e'} \frac{1}{\sqrt{2^n}}\sum_{z \in 2^n} (-1)^{(v+e) \cdot z} \ket{z}\\
        &= \mathsf{H}^{\otimes n} \mathbb{P}_{C^\perp} \frac{(-1)^{e \cdot e'}}{\sqrt{2^n |C|}}\sum_{v \in C} \sum_{z \in 2^n} (-1)^{(v+e) \cdot (z + e')} \ket{z}\\
        &= \mathsf{H}^{\otimes n} \mathbb{P}_{C^\perp} \frac{(-1)^{e \cdot e'}}{\sqrt{2^n |C|}} \sum_{z' \in 2^n} \sum_{v \in C} (-1)^{(v+e) \cdot z'} \ket{z'+e'}\\
        &= \mathsf{H}^{\otimes n} \mathbb{P}_{C^\perp} \frac{(-1)^{e \cdot e'}}{\sqrt{2^n /|C|}} \sum_{z' \in C^\perp} (-1)^{z' \cdot e} \ket{z'+e'}\\
        &= \mathsf{H}^{\otimes n} \frac{(-1)^{e \cdot e'}}{\sqrt{|C^\perp|}}\sum_{z' \in C^\perp} (-1)^{z'\cdot e} \ket{z'+e'}\\
        &= \frac{1}{\sqrt{|C|}}\sum_{v \in C} (-1)^{v\cdot e'} \ket{v+e}
    \end{align*}
where we set $z' \equiv z+e'$ and used the result that if $z' \in C^\perp$ then $\sum_{v \in C}(-1)^{v \cdot z'}=|C|$ and if $z' \notin C^\perp$ then $\sum_{v \in C}(-1)^{v \cdot z'}=0$.

We now show that $V\ket{\psi}=0$ for all $\ket{\psi}$ that are orthogonal to \emph{every} tolerable noisy state $\ket{C_{e,e'}}$. We start by writing $\psi$ as
\begin{equation*}
    \ket{\psi} = \sum_{x \in 2^n} c_x \ket{x}.
\end{equation*}
Then the orthogonality condition gives
\begin{align*}
    0 &= \braket{\psi \, | \, C_{e,e'}}\\
    &= \bigg(\sum_{x \in 2^n} c^*_x \bra{x}, \frac{1}{\sqrt{|C|}} \sum_{v \in C} (-1)^{v \cdot e'} \ket{v+e} \bigg)\\
    \Longleftrightarrow 0 &= \sum_{x \in C+e} c^*_x (-1)^{(x+e) \cdot e'}\\
    \Longleftrightarrow 0 &= \sum_{v \in C} c_{v+e} (-1)^{v \cdot e'}\\
\end{align*}
for all $(e, e') \in \mathcal{E}$. Then,
\begin{align*}
    V \ket{\psi} &= \mathsf{H}^{\otimes n} \mathbb{P}_{C^\perp} \mathsf{H}^{\otimes n} \mathbb{P}_{C} \sum_{x \in 2^n} c_x \ket{x}\\
    %%%
    &= \mathsf{H}^{\otimes n} \mathbb{P}_{C^\perp} \mathsf{H}^{\otimes n} \sum_{e \in \mathcal{E}_X} \bigg(\sum_{v \in C} c_{v+e} \ket{v+e}\bigg)\\
    %%%
    &= \mathsf{H}^{\otimes n} \mathbb{P}_{C^\perp} \sum_{e \in \mathcal{E}_X} \bigg( \sum_{v \in C} c_{v+e} \bigg(\frac{1}{\sqrt{2^n}}\sum_{y \in 2^n} (-1)^{(v+e) \cdot y}\ket{y}\bigg)\bigg)\\
    %%%
    &= \mathsf{H}^{\otimes n} \mathbb{P}_{C^\perp} \sum_{e \in \mathcal{E}_X} \bigg( \sum_{v \in C} c_{v+e} \bigg(\frac{1}{\sqrt{2^n}}\sum_{y' \in 2^n} (-1)^{(v+e) \cdot (y'+e')}\ket{y'+e'}\bigg)\bigg)\\
    %%%
    &= \mathsf{H}^{\otimes n} \mathbb{P}_{C^\perp} \sum_{e \in \mathcal{E}_X} \bigg( \sum_{v \in C} c_{v+e} \bigg(\frac{1}{\sqrt{2^n}}(-1)^{(v+e) \cdot e'}\sum_{y' \in 2^n} (-1)^{(v+e) \cdot y'}\ket{y'+e'}\bigg)\bigg)\\
    %%%
    &= \mathsf{H}^{\otimes n} \sum_{(e, e') \in \mathcal{E}} \bigg( \sum_{v \in C} c_{v+e} \bigg(\frac{1}{\sqrt{2^n}}(-1)^{(v+e) \cdot e'}\sum_{z \in C^\perp} \ket{z+e'} \bigg)\bigg)\\
    %%%
    &= \mathsf{H}^{\otimes n} \sum_{(e, e') \in \mathcal{E}} \bigg(\bigg( \sum_{v \in C} c_{v+e} (-1)^{v \cdot e'} \bigg) \bigg(\frac{(-1)^{e \cdot e'}}{\sqrt{2^n}}\sum_{z \in C^\perp} \ket{z+e'} \bigg)\bigg)\\
    %%%
    &=0
\end{align*}
where $e' \in \mathcal{E}_Z$.

Then $V$ is a sum of projectors, $V=\sum_{(e,e') \in \mathcal{E}}\ket{C_{e,e'}}\bra{C_{e,e'}}$, which is itself a projector. Indeed, $V$ is a projector if and only if $\ket{C_{e,e'}}$ and $\ket{C_{s,s'}}$ are orthogonal whenever \mbox{$(e, e') \neq (s, s')$}.

\begin{align*}
    \braket{C_{e,e'} | C_{s,s'}} = \bigg(\frac{1}{\sqrt{|C|}}\sum_{v \in C} (-1)^{v\cdot e'} \bra{v+e}, \frac{1}{\sqrt{|C|}}\sum_{z \in C} (-1)^{z\cdot s'} \ket{z+s} \bigg)
\end{align*}
The above expression is indeed equal to zero for a fixed $C \in \mathcal{W}$ and when $(e, e') \neq (s, s')$. To see this, first observe that for a fixed $v$, $\braket{v+e \, | \, z + s} =0$ unless $v+e=z+s$ for some $z \in C$, or rather, $z=v+e+s$, but this is impossible for $e \neq s$, since $v,z \in C$ and, because $d(C) \geq 2q+1$, we have $e,s, e+s \notin C$, and hence $z \in C$ while $v+e+s \notin C$. So, suppose that $e=s$ but, by the assumption $(e, e') \neq (s, s')$, we now have $e' \neq s'$. In the Fourier basis (apply the Hadamard transform) we get
\begin{align*}
    \bigg(\frac{(-1)^{e \cdot e'}}{\sqrt{|C^\perp|}}\sum_{v \in C^\perp} (-1)^{v\cdot e} \bra{v+e'}, \frac{(-1)^{e \cdot e'}}{\sqrt{|C^\perp|}}\sum_{z \in C^\perp} (-1)^{z\cdot s} \ket{z+s'} \bigg).
\end{align*}
Again, we get zero unless $z = v+e'+s'$ but $v,z \in C^\perp$ and, because $d(C^\perp) \geq 2q+1$, we have $e', s', e'+s' \notin C^\perp$ and hence $z \in C^\perp$ while $v + e' + s' \notin C^\perp$. Thus, $\ket{C_{e,e'}}$ and $\ket{C_{s,s'}}$ are orthogonal because their error vectors are distinct and because $C \in \mathcal{W}$.
\end{proof}

Now recall that for $q$ arbitrary errors, the error operator $E$ takes the form
\begin{equation}
    E = \sum_{(e,e') \in A} \alpha_{e,e'} X^e Z^{e'},
\end{equation}
where $A$ is some subset of $\mathcal{E}_q$ which contains at least one pair $(e,e')$ that corresponds to an error of weight $q$. Then our noisy money state is of the form
\begin{equation}
\label{eq:arbitrary-noise-money-state}
    E \ket{C} = \sum_{(e,e') \in A} \alpha_{e,e'} \ket{C_{e,e'}}.
\end{equation}
Now suppose we apply the projector $\mathbb{P}_C$ (the process is similar for $\mathbb{P}_{C^\perp}$). That is, we apply the oracle $U_{C_{\mathcal{E}_X}}$ controlled by $\ket{+}$,
\begin{align}
    CU_{C_{\mathcal{E}_X}}(E \ket{C}, \ket{+}) 
    &= \sum_{(e,e') \in A} \alpha_{e,e'} \ket{C_{e,e'}} \ket{-},\\
    &= \left(\sum_{(e,e') \in A} \alpha_{e,e'} \ket{C_{e,e'}}\right) \ket{-},
\end{align}
where the first equality follows by the fact that $\ket{C_{e,e'}} \in C_{\mathcal{E}_X}$ for each $(e,e') \in A$. We will observe $\ket{-}$ when we measure the control qubit, and thus we accept the state. In general, linearity and \cref{lemma:probability-of-acceptance} ensures that a state affected by $q$ arbitrary errors will always be accepted. Notably, this process does not cause a ``collapse" of the arbitrary error as in the typical error correction procedure (see \cref{section:the-stabilizer-formalism}).

We end this section with the following lemma, which we will need to prove \cref{theorem:lower-bound-average-case-counterfeiting}. The need for this lemma is a consequence of the $C \in \mathcal{W}$ condition.

\begin{lem}
\label{lemma:linear-isometry}
    Let $f:\mathbb{F}^n_2 \rightarrow \mathbb{F}^n_2$ be an invertible linear isometry. Then for any $C \in \mathcal{W}$, we have $f(C) \in \mathcal{W}$.
\end{lem}

\begin{proof}
    A classical linear code with distance at least $2q+1$ is able to correct up to $q$ errors. Recall that the distance of a code $C$ is defined as
    \begin{equation*}
        d(C) \coloneqq \min_{x,y \in C, x \neq y} d(x,y)
    \end{equation*}
    where $d(x,y)$ is the Hamming distance between $x$ and $y$. Since $f$ is an isometry (\emph{i.e.}, distance-preserving map) we have $d(f(x), f(y)) = d(x,y)$ for all $x,y \in C$ and hence \newline\mbox{$d(C) = d(f(C))$}. Hence, $f(C)$ has the same distance and is able to correct the same number of errors.
\end{proof}

\subsubsection{The subspace approach}
\label{section:the-subspace-approach}

In this approach to verification, let $C \in \mathcal{W}$, where $\mathcal{W}$ is defined by \cref{eq:applicable-CSS-codes} (\emph{i.e.}, $C$ is an applicable CSS code, as described in \cref{section:direct-error-correction}). Then observe that an error syndrome of all zeros indicates a codeword that has been unaffected by noise (see \cref{remark:zero-syndrome}); put another way, the codeword belongs to the common $+1$ eigenspace of all the generators. If the error syndrome contains, for instance, two $1$'s, then we can interpret this as the banknote state belonging to the intersection of all $+1$ eigenspaces and two $-1$ eigenspaces. From this perspective, each valid syndrome corresponds to a particular intersection of eigenspaces. 

\begin{remark}
\label{remark:zero-syndrome}
    In general, an error syndrome of all zeros indicates that the codeword has been unaffected by noise \emph{or} affected by undetectable noise. To elaborate, let $L$ be a logical operator of a stabilizer code and let $\ket{\psi}$ be a codeword. Recall that because $L$ commutes with all elements of the stabilizer, $L\ket{\psi}$ will also produce a syndrome of all zeros (\emph{i.e.}, $L\ket{\psi}$ is still in the codespace). So, it is not necessarily true that a syndrome of all zeros correlates to a noiseless codeword. However, in the case of an applicable CSS code where there is only a single codeword (\emph{i.e.}, no logical operators), we can say that the all-zero syndrome implies a noiseless codeword.
\end{remark}

Formally, let $\mathfrak{S}$ denote the set of good syndromes (\emph{i.e.}, if $s \in \mathfrak{S}$, then the error that produced $s$ is correctable). Under the stabilizer formalism, every $s \in \mathfrak{S}$ corresponds to a particular intersection of eigenspaces of the generators. Denote such a subspace as $Q_s$, where $s \in \mathfrak{S}$ indexes the subspace. Then for each $s \in \mathfrak{S}$, we want to check membership in $Q_s$.

In the context of a CSS code, the syndrome is a concatenation $s=s^X \Vert s^Z$, where $s^X$ ($s^Z$) denotes the syndrome that arises from bit-flip (phase-flip) errors. So, for a given \mbox{$s \in \mathfrak{S}$}, instead of checking membership in the single subspace $Q_s$, one checks for membership in two subspaces, $Q_{s^X}$ and $Q_{s^Z}$. Let $\mathfrak{S}^X, \mathfrak{S}^Z$ denote, respectively, the set of syndromes that correspond to correctable bit-flip and phase-flip errors. Then let $U_{Q_{\mathfrak{S}^X}}, U_{Q_{\mathfrak{S}^Z}}$ denote the classical oracles that check for membership in the subspaces $Q_{\mathfrak{S}^X}, Q_{\mathfrak{S}^Z}$, respectively, where these subspaces are defined as,
\begin{equation*}
    Q_{\mathfrak{S}^X} \coloneqq \bigcup_{s^X \in \mathfrak{S}^X} Q_{s^X}
    \quad \text{ and } \quad
    Q_{\mathfrak{S}^Z} \coloneqq \bigcup_{s^Z \in \mathfrak{S}^Z} Q^\mathsf{H}_{s^Z}.
\end{equation*}
Note that $Q^{\mathsf{H}}_{s^Z}$ denotes $Q_{s^Z}$ in the Hadamard basis. In a similar manner to the subset approach, the first projector $\mathbb{P}_{C}$ is constructed as follows:
\begin{enumerate}
    \item 
        For the oracle $U_{Q_{\mathfrak{S}^X}}$, we introduce a control qubit $\ket{+}$ and apply the classical oracle controlled on $\ket{+}$.
    \item 
        We then measure the control qubit and accept if the outcome $\ket{-}$ is observed, and reject otherwise.
\end{enumerate}
The process for phase-flip errors is only slightly different: repeat the above steps using $Q_{\mathfrak{S}^Z}$ to construct the second projector $\mathbb{P}_{C^\perp}$. The final verification operator would then be given by $V = \mathsf{H}^{\otimes n} \mathbb{P}_{C^\perp} \mathsf{H}^{\otimes n} \mathbb{P}_{C}$. 

With this $V$, \cref{lemma:probability-of-acceptance} also holds, and for the same reasons as in the subset approach, the verification operator $V$ can be used to accept arbitrary errors. It is worth noting that if we do not use an applicable CSS code, then \cref{lemma:probability-of-acceptance} does not necessarily hold, since we can no longer guarantee that $V$ only accepts a codeword and its tolerated noisy variants. A codeword affected by a logical operator (\emph{i.e.}, a highly corrupted codeword) could be accepted, but this possibility is eliminated by using an applicable CSS code.

\begin{remark}
    The subset approach in \cref{section:the-subset-approach} can be viewed as a generalization of this subspace approach. To see this, we can compare $C_{\mathcal{E}_X}$, as defined in \cref{eq:subset-collections}, with $Q_{\mathfrak{S}^X}$. The set $C_{\mathcal{E}_X}$ is a union of subsets where each subset in the union corresponds to a tolerated bit-flip error; likewise, each subspace in the union $Q_{\mathfrak{S}^X}$ is also a subset which corresponds to tolerating a bit-flip error. Additionally, the verification operator $V$ is constructed in the same way for both cases and corresponds to a projector onto the same subspace. The subspace approach is then an interesting case that closely resembles error correction and offers the possibility of being instantiated in the plain model (see \cref{section:future-directions}).
\end{remark}

To end this section, we remark that the existence of applicable CSS codes, $C \in \mathcal{W}$, is guaranteed by the Gilbert-Varshamov bound for CSS codes, \cref{eq:GV-bound}. In our case, we have $k=0$, and so the bound becomes
\begin{equation}
\label{eq:GV-bound-applicable}
    0 \geq 1-2H\left(\frac{2q}{n}\right).
\end{equation}
If we now fix a value for $q$, we can find an appropriate $n$ such that \cref{eq:GV-bound-applicable} is satisfied. This relationship is visualized in \cref{fig:GV-bound}.

\begin{figure}[ht]
    \centering
    \includegraphics[scale=0.5]{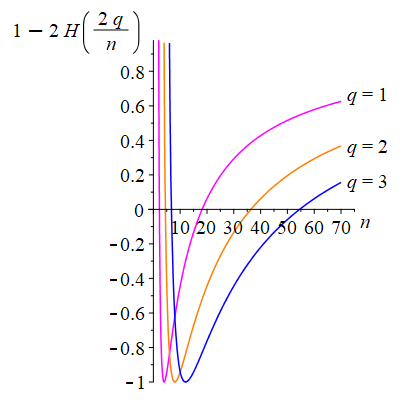}
    \caption{A plot of the right-hand-side of \cref{eq:GV-bound-applicable} for various choices of $q$. For a fixed choice of $q$, \cref{eq:GV-bound-applicable} guarantees the existence of an applicable CSS codes as long as $n$ is chosen such that the corresponding point on the plot is below the horizontal axis.}
    \label{fig:GV-bound}
\end{figure}

As an example of an applicable CSS code, consider the following $[6, 3]$ classical linear code $C$,
\begin{equation*}
    G_C = 
    \begin{pmatrix}
        1 & 0 & 0\\
        0 & 1 & 0\\
        0 & 0 & 1\\
        0 & 1 & 1\\
        1 & 1 & 0\\
        1 & 0 & 1
    \end{pmatrix},
    \quad
    H_C =
    \begin{pmatrix}
        0 & 1 & 1 & 1 & 0 & 0\\
        1 & 1 & 0 & 0 & 1 & 0\\
        1 & 0 & 1 & 0 & 0 & 1
    \end{pmatrix}.
\end{equation*}
Since the distance of this code is $d=3$ and the inequality $d \geq 2q+1$ is satisfied for $q=1$, this code can correct a single error. The dual code $C^\perp$ has the following generator and parity check matrix,
\begin{equation*}
    G_{C^\perp} = H_C^T =
    \begin{pmatrix}
        0 & 1 & 1\\
        1 & 1 & 0\\
        1 & 0 & 1\\
        1 & 0 & 0\\
        0 & 1 & 0\\
        0 & 0 & 1
    \end{pmatrix},
    \quad
    H_{C^\perp} = G_C^T =
    \begin{pmatrix}
        1 & 0 & 0 & 0 & 1 & 1\\
        0 & 1 & 0 & 1 & 1 & 0\\
        0 & 0 & 1 & 1 & 0 & 1
    \end{pmatrix}.
\end{equation*}
The $C^\perp$ code also has a distance of $d=3$ and can thus also correct a single error. Hence, as a CSS code with $C_1, C_2 = C$ we can correct a single arbitrary error. The single codeword of this CSS code is given by
\begin{align*}
    \ket{C} = \frac{1}{\sqrt{8}} 
    \big(
    &\ket{000000}+
    \ket{100011}+
    \ket{010110}+
    \ket{001101}+\\
    &\ket{110101}+
    \ket{101110}+
    \ket{011011}+
    \ket{111000}
    \big).
\end{align*}

\begin{remark}
    Although the existence of applicable CSS codes is easily determined (with the Gilbert-Varshamov bound, \cref{eq:GV-bound-applicable}), explicitly constructing them for arbitrary $n$ and $q$ could pose a difficult challenge. It would be useful, for implementation purposes, to have a method for constructing such codes, though we are unaware of any such techniques. 
\end{remark}

%%%%%%%%%%%%%%%%%%%%%%%%%%%%%%%%%%%%%%%%%%%%%%%%%%%%%%

\section{Security}
\label{section:security}

To prove security, we follow the same approach of \cite{AC12}, changing the proofs as necessary. To start, we look at one of the crucial components from \cite{AC12} that was used for proving security: the amplification of counterfeiters theorem. Intuitively, the theorem shows that a counterfeiter that can copy a banknote with small success can be used to copy a banknote almost perfectly; hence, it suffices to prove security against counterfeiters that can copy almost perfectly. The analogous statement for noisy-projective mini-schemes is given as \cref{theorem:amplify-counterfeiters} below. The proof of this theorem in \cite{AC12} essentially starts with the following equality
\begin{equation}
\label{equation:AC12-prob-acceptance}
    \Pr[V(\rho) = \text{accepts}] = F(\rho, S)^2,
\end{equation}
where $F(\rho,S)$ is the fidelity of $\rho$ with the subspace $S$ that the verification algorithm $V$ projects to. Unlike \cite{AC12}, our space $S$ is not 1-dimensional (see \cref{lemma:probability-of-acceptance}), though we can still rely on the proof in \cite{AC12} to get the same complexity bound. However, we will later see how having a ``larger" $S$ will lead to a trade-off between security and noise-tolerance. 

\begin{theorem}[Amplification of counterfeiters]
\label{theorem:amplify-counterfeiters}
    Let $\mathcal{M}=\left(\textsf{Bank, Ver}\right)$ be a noisy-projective mini-scheme, let $S$ be the subspace that the verification operator $V$ projects to, and let \mbox{$\$ = (s,\rho)$} be a valid banknote (recall \cref{defn:valid-banknotes-for-a-noise-tolerant-scheme}). Suppose there exists a counterfeiter $\ccC$ that produces a counterfeit state $\ccC(\$)$ such that
    \begin{equation*}
        F(\ccC(\$), S^{\otimes 2}) \geq \sqrt{\varepsilon}
    \end{equation*}
    and hence
    \begin{equation*}
        \Pr[V^{\otimes 2}(\ccC(\$)) = \text{accepts}] \geq \varepsilon.
    \end{equation*}
    Then for all $\delta > 0$, there is also a modified counterfeiter $\ccC'$ (depending only on $\varepsilon$ and $\delta$, not $\$$), which makes
    \begin{equation}
    \label{eq:amplified-query-complexity-bound}
        \mathcal{O}\left(\frac{\log 1/\delta}{\sqrt{\varepsilon}(\sqrt{\varepsilon} + \delta^2)}\right)
    \end{equation}
    queries to $\ccC, \ccC^{-1}$, and $V$ and which satisfies
    \begin{equation*}
        F(\ccC'(\$), S^{\otimes 2}) \geq 1-\delta
    \end{equation*}
    which implies
    \begin{equation*}
         \Pr[V^{\otimes 2}(\ccC'(\$)) = \text{accepts}] \geq (1-\delta)^2.
    \end{equation*}
\end{theorem}

\begin{proof}
    We first note that
    \begin{align*}
        \Pr[V(\rho) = \text{accepts}] &= \tr(V \rho)\\
        &= \tr\left(\sum_{(e, e') \in \mathcal{E}} \ket{C_{e,e'}} \bra{C_{e,e'}} \rho \right)\\
        &= \sum_{(e, e') \in \mathcal{E}} \tr\left(\ket{C_{e,e'}} \bra{C_{e,e'}} \rho \right)\\
        &= \sum_{(e,e') \in \mathcal{E}} \bra{C_{e,e'}} \rho \ket{C_{e,e'}}\\
        &= F(\rho, S)^2.
    \end{align*}
    Then, for the case of verifying two banknotes supplied by the counterfeiter $\ccC$, we have,
    \begin{align*}
        \Pr[V^{\otimes 2}(\ccC(\rho)) = \text{accepts}]
        &= \sum_{(e,e'), (t,t') \in \mathcal{E}} \bra{C_{e,e'}} \otimes \bra{C_{t,t'}} \ccC(\rho) \ket{C_{e,e'}} \otimes \ket{C_{t,t'}}\\
        &= \sum_{(e,e'), (t,t') \in \mathcal{E}} F(\ccC(\rho), \ket{C_{e,e'}} \otimes \ket{C_{t,t'}})^2\\
        &= F(\ccC(\rho), S^{\otimes 2})^2
    \end{align*}
    By assumption of the theorem,
    \begin{equation*}
        F(\ccC(\rho), S^{\otimes 2}) \geq \sqrt{\varepsilon}
    \end{equation*}
    and thus we indeed have
    \begin{equation*}
        \Pr[V^{\otimes 2}(\ccC(\rho)) = \text{accepts}] \geq \varepsilon.
    \end{equation*}
    Put simply, one can now apply a fixed-point Grover search with $\ccC(\rho)$ as the initial state and $S^{\otimes 2}$ as the target subspace to produce a state $\rho'$ such that $F(\rho', S^{\otimes 2}) \geq 1-\delta$. As shown in \cite{AC12} (see Section 2.3 and the proof of Theorem 13 for details), this can be done in
    \begin{equation}
    \label{eq:query-bound-in-proof}
        \mathcal{O}\left(\frac{\log 1/\delta}{\sqrt{\varepsilon}(\sqrt{\varepsilon} + \delta^2)}\right)
    \end{equation}
    Grover iterations. Each iteration can be implemented using $\mathcal{O}(1)$ queries to $\ccC, \ccC^{-1}$, and $V$. Thus, the total number of queries to $\ccC, \ccC^{-1}$, and $V$ is also given by \cref{eq:query-bound-in-proof}. Then the probability of $\rho'$ being accepted is
    \begin{align*}
        \Pr[V^{\otimes 2}(\rho') = \text{accepts}] &= \sum_{(e,e'), (t,t') \in \mathcal{E}} F(\rho', \ket{C_{e,e'}} \otimes \ket{C_{t,t'}})^2\\
        &= F(\rho', S^{\otimes 2})^2\\
        &\geq (1-\delta)^2. \qedhere
    \end{align*}
\end{proof}

\subsection{The inner-product adversary method}
\label{section:inner-product-adversary-method}

Put simply, proving security amounts to proving that a counterfeiter, who possesses a copy of $\ket{\psi}$, needs to make exponentially many queries to the verification algorithm to prepare $\ket{\psi_e}\ket{\psi_t}$, where $\ket{\psi_e}$ and $\ket{\psi_t}$ are possibly different noisy copies of $\ket{\psi}$. For simplicity, let us forget about noise-tolerance so that the counterfeiter is now trying to prepare $\ket{\psi}\ket{\psi}$. Suppose that $\ket{\psi}$ and $\ket{\phi}$ are two possible quantum money states. Consider a counterfeiting algorithm $\ccC$ that is executed in parallel to clone $\ket{\psi}$ and $\ket{\phi}$. If, say, $\braket{\psi | \phi} = 1/2$, and if the counterfeiting algorithm $\ccC$ succeeds perfectly (\emph{i.e.}, it maps $\ket{\psi}$ to 
$\ket{\psi}^{\otimes 2}$ and $\ket{\phi}$ to $\ket{\phi}^{\otimes 2}$), then $\bra{\psi}^{\otimes 2}\ket{\phi}^{\otimes 2} = 1/4$, meaning that $\ccC$ must decrease the inner product by at least 1/4. But the analysis of \cite{AC12} shows that the \emph{average} inner product can decrease by at most $1/\textsf{exp}(n)$ from a single query, which results in $2^{\Omega(n)}$ queries being needed. 

This analysis is done using the inner-product adversary method developed in \cite{AC12} (which is an adaptation of the quantum adversary method in \cite{Amb02}). Let us now recall the basic ideas of the method (for further details, see \cite{AC12}). The idea is to get an upper bound on the \emph{progress} a quantum algorithm $\ccC$ can make at distinguishing pairs of oracles, as a result of a single query. To make this a bit more formal, let $\ket{\Psi^U_t}$ be $\ccC$'s state after $t$ queries. For now, assume that $\ket{\Psi^U_0}=\ket{\Psi^V_0}$ for all oracles $U$ and $V$. After the final query $T$, we must have, say, $|\braket{\Psi^U_T | \Psi^V_T}| \leq 1/2$. Thus, if we are able to show that $|\braket{\Psi^U_T | \Psi^V_T}|$ can decrease by at most $\varepsilon$ as the result of a single query, then it follows that $\ccC$ must make $\Omega(1/\varepsilon)$ queries. 

The difficulty, overcome in \cite{AC12}, is that in the quantum money framework, $\ccC$ starts with a valid banknote and hence it is not generally true that $\ket{\Psi^U_0}=\ket{\Psi^V_0}$. This starting point is advantageous for $\ccC$ since it can allow them to do much better in decreasing $|\braket{\Psi^U_T | \Psi^V_T}|$ as a result of a single query. The solution in \cite{AC12} is to choose a distribution $\mathcal{D}$ over oracle pairs $(U,V)$ and then analyze the expected inner product
\begin{equation*}
    \mathop{\mathbb{E}}_{(U,V) \sim \mathcal{D}} \left[|\braket{\Psi^U_t | \Psi^V_t}|\right]
\end{equation*}
and how much it can decrease as the result of a single query to $U$ or $V$. Analysis shows that, although $\ccC$ can succeed in significantly reducing $|\braket{\Psi^U_T | \Psi^V_T}|$ for some oracle pairs $(U,V)$, they cannot do so for most pairs.

\subsubsection{Formal description}

This section, which gives the formal description of the inner-product adversary method, is taken almost directly from \cite{AC12}. Let $\mathcal{O}$ denote a set of quantum oracles acting on $n$ qubits each. For each $U \in \mathcal{O}$, assume there exists a subspace $S_U \subseteq \mathbb{C}^{2n}$ such that
\begin{enumerate}[(i)]
    \item $U \ket{\psi} = -\ket{\psi}$ for all $\ket{\psi} \in S_U$
    \item $U \ket{\eta}$ for all $\ket{\eta} \in S^\perp_U$.
\end{enumerate}
Let $R \subset \mathcal{O} \times \mathcal{O}$ be a symmetric binary relation on $\mathcal{O}$, with the properties that
\begin{enumerate}[(i)]
    \item $(U,U) \notin R$ for all $U \in \mathcal{O}$, and
    \item for every $U \in \mathcal{O}$ there exists a $V \in \mathcal{O}$ such that $(U,V) \in R$.
\end{enumerate}
Suppose that for all $U \in \mathcal{O}$ and all $\ket{\eta} \in S^\perp_U$, we have
\begin{equation*}
    \mathop{\mathbb{E}}_{V : (U,V) \in R} \Big[F(\ket{\eta}, S_V)^2\Big] \leq \varepsilon,
\end{equation*}
where $F(\ket{\eta}, S_V) = \max_{\ket{\psi} \in S_V} |\braket{\eta|\psi}|$ is the fidelity between $\ket{\eta}$ and $S_V$. Let $\ccC$ be a quantum oracle algorithm, and let $\ccC^U$ denote $\ccC$ run with the oracle $U \in \mathcal{O}$. Suppose $\ccC^U$ begins in the state $\ket{\Psi^U_0}$ (possibly dependent on $U$). Let $\ket{\Psi^U_t}$ denote the state of $\ccC^U$ immediately after the $t^{th}$ query. Also, define a \emph{progress measure} $p_t$ by
\begin{equation*}
    p_t \coloneqq \mathop{\mathbb{E}}_{V : (U,V) \in R} \Big[|\braket{\Psi^U_t | \Psi^V_t}|\Big].
\end{equation*}

The following lemma bounds how much $p_t$ can decrease as the result of a single query. See \cite{AC12} for a proof.

\begin{lem}[Bound on progress rate \cite{AC12}]
\label{lemma:bound-on-progress-rate}
    \begin{equation*}
        p_t \geq p_{t-1} - 4\sqrt{\varepsilon}.
    \end{equation*}
\end{lem}

As a result of \cref{lemma:bound-on-progress-rate}, we have
\begin{lem}[Inner-product adversary method \cite{AC12}]
\label{lemma:inner-product-adversary-method}
    Suppose that we initially have \mbox{$|\braket{\Psi^U_0 | \Psi^V_0}| \geq c$} for $(U,V) \in R$, whereas by the end we need $|\braket{\Psi^U_T | \Psi^V_T}| \leq d$ for all $(U,V) \in R$. Then $\ccC$ must make 
    \begin{equation*}
        T = \Omega\bigg(\frac{c-d}{\sqrt{\varepsilon}}\bigg)
    \end{equation*}
    oracle queries.
\end{lem}

\subsubsection{Applying the method}
\label{section:applying-the-method}

For much of the security proof, we suppose that the counterfeiter only has access to the oracles $(U_{C_{\mathcal{E}_X}}, U_{C_{\mathcal{E}_Z}})$ as opposed to the oracles in the noisy-projective mini-scheme specified in \cref{section:formal-specification}. The generalization to the full scheme will be done later in \cref{theorem:security-mini-scheme}.

Our starting point for the analysis will be to think of the pair of oracles $(U_{C_{\mathcal{E}_X}}, U_{C_{\mathcal{E}_Z}})$ as a single oracle. In the case of \cite{AC12}, they have the pair of oracles $(U_A, U_{A^\perp})$ and the corresponding single oracle is constructed in the following way. First, they consider a subset $A^* \subset \{0,1\}^{n+1}$ defined by
\begin{equation*}
    A^* \coloneqq (0,A) \cup (1, A^\perp).
\end{equation*}
Then they let $S_{A^*}$ denote the subspace of $\mathbb{C}^{2n+1}$ that is spanned by basis states $\ket{x}$ such that $x \in A^*$. This allows their collection of oracles $(U_A, U_{A^\perp})$ to be thought of as a single oracle $U_{A^*}$ which satisfies
\begin{equation*}
    U_{A^*}\ket{\psi}
    =
    \begin{cases}
        -\ket{\psi}, \quad &\text{if} \,\, \ket{\psi} \in S_{A^*}\\
        \ket{\psi}, \quad &\text{if} \,\, \ket{\psi} \in S^\perp_{A^*}
    \end{cases}.
\end{equation*}

To generalize this to our setting, first suppose that our scheme only tolerates the errors $\mathcal{E} = \{(e,e') , (t,t')\}$. Then the subset $C^* \subset \{0,1\}^{n+2}$, defined by
\begin{equation}
\label{eq:enlarged-bit-flips-subspace-example}
    C^* \coloneqq (00, C+e) \cup (01, C^\perp+e') \cup (10, C+t) \cup (11, C^\perp+t'),
\end{equation}
would be a generalization of $A^*$. Let us also introduce the following notation which naturally generalizes for a larger set of error vectors,
\begin{align*}
    C^*_{\mathcal{E}_X} &\coloneqq (00, C+e) \cup (10, C+t)\\
    C^*_{\mathcal{E}_Z} &\coloneqq (01, C^\perp+e') \cup (11, C^\perp+t').
\end{align*}
In general, every error vector from $\mathcal{E}_X$ and $\mathcal{E}_Z$ corresponds to a term in the union that defines $C^*$. So, we would have $2|\mathcal{E}_X|$ (\emph{i.e.}, $|\mathcal{E}_X|+|\mathcal{E}_Z|$) terms in the union, meaning $C^* \subset \{0,1\}^{n+k}$ where $k$ is such that
\begin{equation}
\label{equation:k-for-C*}
    2^k = 2|\mathcal{E}_X|
    \Leftrightarrow
    k = 1 + \log(|\mathcal{E}_X|)
\end{equation}
Now let $S_{C^*}$ be the subspace of $\mathbb{C}^{2^{n+k}}$ that is spanned by basis states $\ket{x}$ such that $x \in C^*$. Then our pair of oracles $(U_{C_{\mathcal{E}_X}}, U_{C_{\mathcal{E}_Z}})$ corresponds to the single oracle $U_{C^*}$ which satisfies
\begin{equation}
\label{eqn:single-classical-oracle}
    U_{C^*}\ket{\psi}
    =
    \begin{cases}
        -\ket{\psi}, \quad &\text{if} \,\, \ket{\psi} \in S_{C^*}\\
        \ket{\psi}, \quad &\text{if} \,\, \ket{\psi} \in S^\perp_{C^*}
    \end{cases}.
\end{equation}
A notable consequence of this generalization to the noisy setting is that a single query to $U_{C_{\mathcal{E}_X}}$ equates to $|\mathcal{E}_X|$ queries to $U_{C^*}$, and likewise for $U_{C_{\mathcal{E}_Z}}$. To see this, observe the example of $C^*$ in \cref{eq:enlarged-bit-flips-subspace-example} where a query to $U_{C_{\mathcal{E}_X}}$ is achieved by preparing ancilla qubits set to $\ket{00}$ and $\ket{10}$ and then querying $U_{C^*}$ with each of these ancilla qubits. For security, this means that the lower bound for the number of queries to $(U_{C_{\mathcal{E}_X}}, U_{C_{\mathcal{E}_Z}})$ is the lower bound for queries to $U_{C^*}$ divided by $2|\mathcal{E}_X|$. We return to this point in \cref{corollary:single-oracle-to-oracle-pair}.

The following theorem shows that, for the case of perfect counterfeiting, exponentially many queries are needed. In \cite{AC12}, perfect counterfeiting meant that the counterfeiter, who possessed a copy of $\ket{C}$, had to prepare $\ket{C} \otimes \ket{C}$. In our noise-tolerant setting, the counterfeiter has more flexibility by only needing to prepare noisy copies of $\ket{C}$. This difference adds some extra steps to the proof but the lower bound on the needed queries remains the same as the analogous result in \cite{AC12}.

\begin{theorem}[Lower bound for perfect counterfeiting]
\label{theorem:lower-bound-perfect-counterfeit}
    Given one copy of $\ket{C}$, as well as oracle access to $U_{C^*}$, a counterfeiter needs $\Omega\left(2^{n/4}\right)$ queries to prepare $\ket{C_{e,e'}} \otimes \ket{C_{t,t'}}$, where $(e,e'), (t,t') \in \mathcal{E}$, with certainty (for a worst-case $\ket{C}$).    
\end{theorem}

\begin{proof}
    Let the set $\mathcal{O}$ contain $U_{C^*}$ for every possible subspace $C \subseteq \mathbb{F}^n_2$ such that \mbox{$\dim (C) = n/2$}. Let $(U_{C^*}, U_{D^*}) \in R$ if and only if $\dim(C \cap D) = n/2-1$. Then given $U_{C^*} \in \mathcal{O}$ and $\ket{\eta} \in S^\perp_{C^*}$, let
    \begin{equation*}
        \ket{\eta} = \sum_{x \in \{0,1\}^{n+k} \setminus C^*} \alpha_x \ket{x}.
    \end{equation*}
    We have
    {
    \allowdisplaybreaks
    \begin{align}
        \mathop{\mathbb{E}}_{U_{D^*} \, : \, (U_{C^*},U_{D^*}) \in R} \left[F(\ket{\eta}, S_{D^*})^2\right] 
        &= \mathop{\mathbb{E}}_{D \, : \, \dim(D)=n/2, \, \dim(C \cap D) = n/2-1} \left[\sum_{x \in D^* \setminus C^*} |\alpha_x|^2 \right] \label{first-line} \\
        &\leq \max_{x \in \{0,1\}^{n+k} \setminus C^*} \bigg( \Pr_{D \, : \, \dim(D)=n/2, \, \dim(C \cap D) = n/2-1} [x \in D^*] \bigg) \label{second-line} \\
        &= \max_{x \in \{0,1\}^{n+k-1} \setminus C^*_{\mathcal{E}_X}} \bigg( \Pr_{D \, : \, \dim(D)=n/2, \, \dim(C \cap D) = n/2-1} [x \in D^*_{\mathcal{E}_X}] \bigg) \label{third-line} \\
        &= \frac{|D^*_{\mathcal{E}_X} \setminus C^*_{\mathcal{E}_X}|}{|\{0,1\}^{n+k-1} \setminus C^*_{\mathcal{E}_X}|} \label{fourth-line}\\
        &= \frac{|D^*_{\mathcal{E}_X} \setminus C^*_{\mathcal{E}_X}|}{|\mathcal{E}_X|2^n - |\mathcal{E}_X|2^{n/2}} \label{fifth-line}\\
        & \leq \frac{|\mathcal{E}_X|2^{n/2} - 2^{n/2-1}}{|\mathcal{E}_X|2^n - |\mathcal{E}_X|2^{n/2}} \label{sixth-line}\\
        &= \frac{|\mathcal{E}_X|-2^{-1}}{|\mathcal{E}_X|\left(2^{n/2} - 1\right)} \label{seventh-line} \\
        &\leq \frac{1}{2^{n/2} - 1} \label{eighth-line} \\
        &\leq \frac{1}{2^{n/2-1}} \label{ninth-line}
    \end{align}
    }
    \Cref{first-line} uses the definition of the fidelity. As done in \cite{AC12}, \cref{second-line} uses the easy direction of the minimax theorem. \Cref{third-line} uses the symmetry between $C^*_{\mathcal{E}_X}$ and $C^*_{\mathcal{E}_Z}$. \Cref{fifth-line} uses \cref{equation:k-for-C*} and the following facts: each coset has the same number of elements as the group; two cosets for the same group are either disjoint or identical; and that there is a coset for each error vector. \Cref{sixth-line} uses the fact that $|D^*_{\mathcal{E}_X} \setminus C^*_{\mathcal{E}_X}| = |D^*_{\mathcal{E}_X}| - |D^*_{\mathcal{E}_X} \cap C^*_{\mathcal{E}_X}|$ and that the intersection of two cosets from $D$ and $C$ are either empty or a coset of $D \cap C$; while we do not know how many distinct coset intersections we have, we know that there is at least one such intersection: the trivial coset $C \cap D$ which has $2^{n/2-1}$ elements. The conclusion here is that $\varepsilon \coloneqq 2^{-n/2+1}$.

    Now fix $(U_{C^*}, U_{D^*}) \in R$. Then initially $|\braket{C \, | \, D}| = 1/2$. The counterfeiter must map $\ket{C}$ to some state $\ket{f_C} \coloneqq \ket{C_{e,e'}} \ket{C_{t,t'}} \ket{\text{garbage}_C}$, where $(e,e'), (t,t') \in \mathcal{E}$, and likewise for $\ket{D}$. Thus, $|\braket{f_C | f_D}| \leq 1/4$. To see this, observe the following:
    \begin{align*}
        \braket{C_{e,e'} | D_{t,t'}} &= \bigg(\frac{1}{\sqrt{|C|}}\sum_{v \in C} (-1)^{v\cdot e'} \bra{v+e}, \frac{1}{\sqrt{|D|}}\sum_{z \in D} (-1)^{z\cdot t'} \ket{z+t} \bigg)\\
        &= \frac{1}{2^{n/2}} \bigg(\sum_{v \in C} (-1)^{v\cdot e'} \bra{v+e}, \sum_{z \in D} (-1)^{z\cdot t
        '} \ket{z+t} \bigg)\\
        &\leq \frac{1}{2^{n/2}} \sum_{v \in C} \sum_{z \in D} \braket{v+e | z+t}
    \end{align*}
    Now suppose that the intersection of the two cosets, $C + e \cap D + t$ is non-empty. Then we can interpret the above sum of inner-products as the number of elements in a coset of $C \cap D$, which is equal to $|C \cap D| = 2^{n/2-1}$. Hence,
    \begin{equation*}
        \braket{C_{e,e'} | D_{t,t'}} \leq \frac{2^{n/2-1}}{2^{n/2}} = \frac{1}{2},
    \end{equation*}
    which indeed implies $|\braket{f_C | f_D}| \leq 1/4$.

    Then by \cref{lemma:inner-product-adversary-method}, with $c=1/2, d=1/4,$ and $\varepsilon=2^{-n/2+1}$,
    \begin{equation*}
        \Omega
        \left(
        \frac{c-d}{\sqrt{\varepsilon}}
        \right) 
        = \Omega
        \left(
        2^{\left(n-2\right)/4}
        \right)
        = \Omega
        \left(
        2^{n/4}
        \right). \qedhere
    \end{equation*}
\end{proof}

%%%

\begin{remark}
    It would be practically interesting to study the situation where the counterfeiter starts with an \emph{imperfect} copy of $\ket{C}$. This could model the scenario where the initial transfer from the bank is done through a noisy channel, and so a counterfeiter would start with a money state that has already incurred errors. Intuitively, this difference would not result in a substantial change to the analysis. To formally analyze this case, one could apply the inner product adversary method, as done in \cref{theorem:lower-bound-perfect-counterfeit}, but the ``starting point" of $|\braket{C | D}|=1/2$ would now be different on account of the initial noise.
\end{remark}

The following corollary shows that, for the case where the counterfeiter succeeds almost perfectly, exponentially many queries are still needed.

\begin{corollary}[Lower bound for small-error counterfeiting]
\label{cor:lower-bound-small-error-counterfeit}
    Given one copy of $\ket{C}$, as well as oracle access to $U_{C^*}$, a counterfeiter needs $\Omega(2^{n/4})$ queries to prepare a state $\rho$ such that 
    \begin{equation*}
        \Big(\bra{C_{e,e'}} \otimes \bra{C_{t,t'}}\Big) \rho \Big(\ket{C_{e,e'}} \otimes \ket{C_{t,t'}}\Big) \geq 1-\epsilon
    \end{equation*}
    for small $\epsilon$ (say $\epsilon=0.0001$) and some $(e,e'), (t,t') \in \mathcal{E}$, for a worst case $\ket{C}$. The probability that this state passes verification is
    \begin{equation*}
        \Pr[V^{\otimes 2}(\rho) = \text{accepts}] \geq 1-\epsilon.
    \end{equation*}
\end{corollary}

\begin{proof}
    Let $|\braket{C | D}| = c$. If the counterfeiter succeeds, it must map $\ket{C}$ to some state $\rho_C$ such that
    \begin{equation*}
        \Big(\bra{C_{\tilde{e},\tilde{e}'}} \otimes \bra{C_{\tilde{t},\tilde{t}'}}\Big) \rho_C \Big(\ket{C_{\tilde{e},\tilde{e}'}} \otimes \ket{C_{\tilde{t},\tilde{t}'}}\Big) \geq 1-\epsilon
    \end{equation*}
    for some $(\tilde{e},\tilde{e}'),(\tilde{t},\tilde{t}') \in \mathcal{E}$. Likewise, it must map $\ket{D}$ to some state $\rho_D$ such that
    \begin{equation*}
        \Big(\bra{C_{\tilde{s},\tilde{s}'}} \otimes \bra{C_{\tilde{r},\tilde{r}'}}\Big) \rho_D \Big(\ket{C_{\tilde{s},\tilde{s}'}} \otimes \ket{C_{\tilde{r},\tilde{r}'}}\Big) \geq 1-\epsilon
    \end{equation*}
    for some $(\tilde{s},\tilde{s}'),(\tilde{r},\tilde{r}') \in \mathcal{E}$

    Let $\ket{f_C}, \ket{f_D}$ be purifications of $\rho_C, \rho_D$ respectively. Then by \cref{lemma:triangle-inequality-fidelity}
    \begin{align*}
        |\braket{f_C | f_D}| &\leq F(\rho_C, \rho_D)\\
        &\leq \Bigg| \bra{C_{\tilde{e},\tilde{e}'}} \otimes \bra{C_{\tilde{t},\tilde{t}'}} \ket{D_{s,s'}} \otimes \ket{D_{r,r'}} \bigg| + 2\epsilon^{1/4}\\
        &\leq c^2 + 2\epsilon^{1/4}
    \end{align*}
    where the upper bound is achieved by assuming that the above coset intersections are non-empty.
    Thus, with $d \coloneqq c^2 + 2\epsilon^{1/4}$ in \cref{lemma:inner-product-adversary-method},
    \begin{align*}
        \Omega\bigg(\frac{c-d}{\sqrt{\varepsilon}}\bigg) = \Omega\left(\left(c-c^2-2\epsilon^{1/4}\right)\left(2^{\left(n-2\right)/4}\right)\right).
    \end{align*}
    Fixing $c=1/2$ gives $\Omega\left(2^{n/4}\right)$

    The probability that such a state $\rho$ passes verification is
    \begin{align*}
        \Pr[V^{\otimes 2}(\rho) = \text{accepts}] &= \sum_{(e,e')', (t,t') \in \mathcal{E}} \bra{C_{e,e'}} \otimes \bra{C_{t,t'}} \rho \ket{C_{e,e'}} \otimes \ket{C_{t,t'}}\\
        &\geq \left(\bra{C_{\tilde{e},\tilde{e}'}} \otimes \bra{C_{\tilde{t},\tilde{t}'}}\Big) \rho \Big(\ket{C_{\tilde{e},\tilde{e}'}} \otimes \ket{C_{\tilde{t},\tilde{t}'}}\right)\\
        &\geq 1-\epsilon. \qedhere
    \end{align*}
\end{proof}

%%%

Recall the implication of \cref{theorem:amplify-counterfeiters}: since a counterfeiter with small success probability can be amplified to one with high success probability, it suffices to prove security against the latter. We now formally use this theorem, and the previous result on counterfeiters that succeed with high probability (\cref{cor:lower-bound-small-error-counterfeit}), to provide a lower bound to the number of queries needed by a counterfeiter that has small success probability.

While the proof technique for this result is practically the same as the analogous result in \cite{AC12}, our result contains a noticeable difference. Since the space $S$ which our verification operator projects to is not 1-dimensional, the probability of a counterfeiter's state being accepted can be notably higher in comparison to the noise-free setting, even if that state has a low fidelity with the tensor product of two tolerable noisy states. This difference marks the trade-off between security and noise-tolerance.

\begin{corollary}[Lower bound for high-error counterfeiting]
\label{cor:lower-bound-high-error-counterfeit}
    Let $1/\varepsilon = o(2^{n/2})$. Given one copy of $\ket{C}$, as well as oracle access to $U_{C^*}$, a counterfeiter needs $\Omega\left(\sqrt{\varepsilon} 2^{n/4}\right)$ queries to prepare a state $\rho$ such that
    \begin{equation*}
        \max_{(e,e'), (t,t') \in \mathcal{E}}\Big(\bra{C_{e,e'}} \otimes \bra{C_{t,t'}}\Big) \rho \Big(\ket{C_{e,e'}} \otimes \ket{C_{t,t'}}\Big) \geq \varepsilon
    \end{equation*}
    for a worst-case $\ket{C}$. Furthermore, the greatest lower bound on the probability of accepting this state is
    \begin{equation*}
        \Pr[V^{\otimes 2}(\rho) = \text{accepts}] \geq |\mathcal{E}|^2 \varepsilon.
    \end{equation*}
\end{corollary}

\begin{proof}
    Suppose we have a counterfeiter $\ccC$ that makes $o\left(\sqrt{\varepsilon}2^{n/4}\right)$ queries to $U_{C^*}$, and prepares a state $\sigma$ such that
    \begin{align*}
        \max_{(e,e'), (t,t') \in \mathcal{E}}\Big(\bra{C_{e,e'}} \otimes \bra{C_{t,t'}}\Big) &\sigma \Big(\ket{C_{e,e'}} \otimes \ket{C_{t,t'}}\Big) \geq \varepsilon\\
        \Longrightarrow F(\sigma, S^{\otimes 2})^2 &\geq \varepsilon.
    \end{align*}
    Then
    \begin{align*}
        \Pr[V^{\otimes 2}(\sigma) = \text{accepts}] = F(\sigma, S^{\otimes 2})^2 &\geq \varepsilon.
    \end{align*}
    Let $\delta$ be small, say $\delta=0.00001$. Then by \cref{theorem:amplify-counterfeiters}, there exists an amplified counterfeiter $\ccC'$ that makes
    \begin{equation*}
        \mathcal{O}\left(\frac{\log 1/\delta}{\sqrt{\varepsilon}\left(\sqrt{\varepsilon} + \delta^2\right)}\right) = \mathcal{O}\left(\frac{1}{\sqrt{\varepsilon}}\right)
    \end{equation*}
    calls to $\ccC$ and $V_C$, and that prepares a state $\rho$ such that $F(\rho, S^{\otimes 2}) \geq 1-\delta$ which implies
    \begin{align*}
        \Pr[V^{\otimes 2}(\rho) = \text{accepts}] \geq (1-\delta)^2.
    \end{align*}

    Counting the $o\left(\sqrt{\varepsilon}2^{n/4}\right)$ queries from each use of $\ccC$, and $\mathcal{O}(1)$ queries from each use of $V_C$, the total number of queries that $\ccC'$ makes to $U_{C^*}$ is
    \begin{equation*}
        \left[o\left(\sqrt{\varepsilon}2^{n/4}\right) + \mathcal{O}(1)\right] \cdot \mathcal{O}\left(\frac{1}{\sqrt{\varepsilon}}\right) = o\left(2^{n/4}\right).
    \end{equation*}
    But this contradicts \cref{cor:lower-bound-small-error-counterfeit}.

    Now suppose a counterfeiter uses $\Omega\left(\sqrt{\varepsilon}2^{n/4}\right)$ queries to prepare a state $\rho$ such that
    \begin{equation*}
        \max_{(e,e'), (t,t') \in \mathcal{E}}\Big(\bra{C_{e,e'}} \otimes \bra{C_{t,t'}}\Big) \rho \Big(\ket{C_{e,e'}} \otimes \ket{C_{t,t'}}\Big) \geq \varepsilon.
    \end{equation*}
    If $\ccC$ submits this state for verification, their probability of acceptance will be higher if
    \begin{equation*}
        \Big(\bra{C_{e,e'}} \otimes \bra{C_{t,t'}}\Big) \rho \Big(\ket{C_{e,e'}} \otimes \ket{C_{t,t'}}\Big) > 0,
    \end{equation*}
    for all $(e,e'), (t,t') \in \mathcal{E}$. The best scenario for $\ccC$ is that each term obtains the maximum value, in which case we have
    \begin{align*}
        \Pr[V^{\otimes 2}(\rho) = \text{accepts}] &= \sum_{(e,e'), (t,t') \in \mathcal{E}} \bra{C_{e,e'}} \otimes \bra{C_{t,t'}} \rho \ket{C_{e,e'}} \otimes \ket{C_{t,t'}}\\
        &= |\mathcal{E}|^2 \cdot \max_{(e,e'), (t,t') \in \mathcal{E}}\Big(\bra{C_{e,e'}} \otimes \bra{C_{t,t'}}\Big) \rho \Big(\ket{C_{e,e'}} \otimes \ket{C_{t,t'}}\Big)\\
        &\geq |\mathcal{E}|^2 \varepsilon. \qedhere
    \end{align*}
\end{proof}

The following theorem shows that it suffices for a bank to generate uniformly random banknotes. Indeed, as the proof shows, if a counterfeiter is able to clone a uniformly random banknote, then they can clone any banknote. Note that in the noise-free case of \cite{AC12}, the proof of this theorem used an invertible linear map $f$; in our case, we additionally require $f$ to be an isometry (\emph{i.e.}, a distance-preserving map) to ensure that the noise-tolerance is, in a sense, ``preserved" under the mapping $f$.

%%%
\begin{theorem}[Lower bound for average-case counterfeiting]
\label{theorem:lower-bound-average-case-counterfeiting}
    Let $C$ be a uniformly random element of $\mathcal{W}$. Then given one copy of $\ket{C}$, as well as oracle access to $U_{C^*}$, a counterfeiter $\ccC$ needs $\Omega(\sqrt{\varepsilon}2^{n/4})$ queries to prepare a $2n$-qubit state $\rho$ that $V^{\otimes 2}_C$ accepts with probability at least $|\mathcal{E}|^2\varepsilon$, for all $1/\varepsilon=o(2^{n/2})$. Here the probability is taken over the choice $C$, as well as the behaviour of $\ccC$ and $V^{\otimes 2}_C$.
\end{theorem}

\begin{proof}
    Suppose we had a counterfeiter $\ccC$ that violated the above. Using $\ccC$ as a black box, we construct a new counterfeiter $\ccC'$ that violates \cref{cor:lower-bound-high-error-counterfeit}.

    Given a (deterministically-chosen) banknote $\ket{C}$ and oracle access to $U_{C^*}$, first choose an invertible linear isometry $f: \mathbb{F}^n_2 \rightarrow \mathbb{F}^n_2$ uniformly at random. Then $f(C)$ is a uniformly random element of $\mathcal{W}$ by \cref{lemma:linear-isometry}. Furthermore, the state $\ket{C}$ can be transformed into $\ket{f(C)}$ straightforwardly.

    Using the oracle $U_{C^*}$ and the map $f$, we can simulate the corresponding oracle for the money state $\ket{f(C)}$. Let $U_f$ be the unitary that acts as $U_f\ket{x} = \ket{f(x)}$. Then the corresponding oracle is given by $(I^k \otimes U_f) U_{C^*} (I^k \otimes U^\dag_f)$. To see that this works, observe that,
    \begin{align*}
        x \in C_{\mathcal{E}_X} 
        &\Leftrightarrow
        x \in C+e, \text{ for some } e \in \mathcal{E}_X\\
        &\Leftrightarrow
        f(x) \in f(C)+f(e), \text{ where } f(e) \in \mathcal{E}_X\\
        &\Leftrightarrow
        f(x) \in f(C)_{\mathcal{E}_X}
    \end{align*}
    where $f(e) \in \mathcal{E}_X$ follows because $f$ is an isometry and thus preserves the weight of $e$.

    So, by using $\ccC$ for $n/2$-dimensional states that are drawn uniformly randomly from $\mathcal{W}$, the adversary $\ccC'$ can produce a state $\rho_f$ that $V^{\otimes 2}_{f(C)}$ accepts with probability at least $|\mathcal{E}|^2\varepsilon$. Applying $f^{-1}$ to both registers of $\rho_f$, $\ccC'$ can obtain a state $\rho$ that $V^{\otimes 2}_C$ accepts with probability at least $|\mathcal{E}|^2\varepsilon$, thereby contradicting \cref{cor:lower-bound-high-error-counterfeit}.
\end{proof}

We now rephrase \cref{theorem:lower-bound-average-case-counterfeiting} in terms of the oracle pair $(U_{C_{\mathcal{E}_X}}, U_{C_{\mathcal{E}_Z}})$ as opposed to the single oracle $U_{C^*}$. In the noise-free case of \cite{AC12}, this is immediate as their single oracle $U_{A^*}$ exactly corresponds to the pair $(U_A, U_{A^\perp})$. In our case, we recall (see \cref{section:applying-the-method}) that the query complexity for the pair $(U_{C_{\mathcal{E}_X}}, U_{C_{\mathcal{E}_Z}})$ is the query complexity for $U_{C^*}$ divided by $2|\mathcal{E}_X|$. Under this correspondence, we get the following corollary of \cref{theorem:lower-bound-average-case-counterfeiting}.

\begin{corollary}
\label{corollary:single-oracle-to-oracle-pair}
    Let $C$ be a uniformly random element of $\mathcal{W}$. Then given one copy of $\ket{C}$, as well as oracle access to $(U_{C_{\mathcal{E}_X}}, U_{C_{\mathcal{E}_Z}})$, a counterfeiter $\ccC$ needs $\Omega(\sqrt{\varepsilon}2^{n/4}/|\mathcal{E}_X|)$ queries to prepare a $2n$-qubit state $\rho$ that $V^{\otimes 2}_C$ accepts with probability at least $|\mathcal{E}|^2\varepsilon$, for all $1/\varepsilon=o(2^{n/2})$. Here the probability is taken over the choice $C$, as well as the behaviour of $\ccC$ and $V^{\otimes 2}_C$.
\end{corollary}

In the noise-free case, \cref{corollary:single-oracle-to-oracle-pair} implies that $\varepsilon$ is $1/\textsf{exp}(n)$, which constitutes the soundness error in this case. In the noisy case, we need $|\mathcal{E}|^2$ to be constant in $n$ to get the same soundness error, \emph{i.e.}, for $|\mathcal{E}|^2 \varepsilon$ to be $1/\textsf{exp}(n)$. More errors can be tolerated at the cost of increasing the soundness error, though we remark that it is not possible to tolerate too many errors. Recall from \cref{eq:number-of-errors} that the size of $|\mathcal{E}_q|$ is given by
\begin{equation*}
    |\mathcal{E}_q|
    = 
    |\mathcal{E}_X| \cdot |\mathcal{E}_Z|
    =
    \left(
    \sum^q_{j=0} \binom{n}{j}
    \right)^2.
\end{equation*}
When $q$ is large (close to $n$), the sum in the parentheses approaches exponential, and so the soundness error approaches $1$. Intuitively, this makes sense since a high noise-tolerance allows highly imperfect copies to be accepted, which is advantageous for a counterfeiter. In addition to increasing the soundness error, \cref{corollary:single-oracle-to-oracle-pair} shows that tolerating more errors has the effect of reducing the number of queries needed by an adversary. On the other hand, when $q=kn$, with $0<k\leq 1/2$, we have the following bound from $\cite{FG06}$,
\begin{equation}
\label{eq:error-bound}
    |\mathcal{E}_q| \leq 2^{2 H(k) n},
\end{equation}
where $H(\cdot)$ is the binary Shannon entropy. Then for a small number of arbitrary errors $q$ (\emph{i.e.}, small $k$), we see that $|\mathcal{E}_q|$ stays small enough to ensure that we can get a reasonably low soundness error.

This trade-off between noise-tolerance and soundness is visualized in \cref{fig:soundness-error-trade-off}. We see that if we fix a desirable value for the soundness error, then increasing $q$ necessitates an increase of $n$ to maintain the same soundness error. Given this, it is interesting to consider what an optimal choice of $n$ would be for a fixed soundness error. If, for example, we consider a soundness error of 0.03, inspection of \cref{fig:soundness-error-trade-off} reveals that the ratio $q/n$ is approximately $1.75\%$ for $(n,q)=(57,1)$, but $1.81\%$ for $(n,q)=(166,3)$, and thus increasing $n$ can improve the noise-tolerance. For larger soundness error, it is less clear that this still holds true. We discuss this further in \cref{section:future-directions} as a direction for future research.

\begin{figure}[ht]
    \centering
    \includegraphics[scale=0.5]{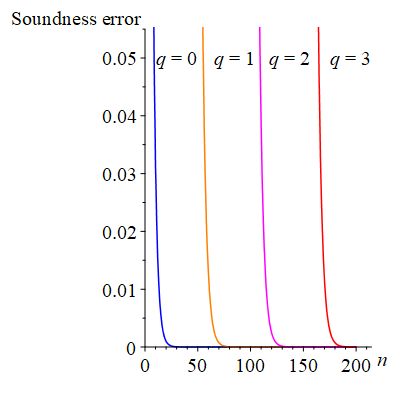}
    \caption{A plot of the soundness error $|\mathcal{E}_q|^2\varepsilon$ for various choices of $q$ (the number of tolerated arbitrary errors) against the number of qubits $n$. Here, we take $\varepsilon=2^{-n/2}$.}
    \label{fig:soundness-error-trade-off}
\end{figure}

We now return to proving security for the noisy-projective mini-scheme $\mathcal{M}$ specified in \cref{section:formal-specification}. To do this, we use the same proof technique as the analogous result in \cite{AC12}; the notable difference here is that we must account for the fact that the formal specification of our scheme uses conjugate coding states, though this does not affect security.

\begin{theorem}[Security of the noisy-projective mini-scheme]
\label{theorem:security-mini-scheme}
    The noisy-projective mini-scheme $\mathcal{M} = (\textsf{Bank}_\mathcal{M}, \textsf{Ver}_\mathcal{M})$, which is defined relative to the classical oracle $U$, has perfect completeness and $|\mathcal{E}|^2/\textsf{exp}(n)$ soundness error.
\end{theorem}

\begin{proof}
    Perfect completeness follows from how $\mathcal{M}$ is defined (see \cref{section:formal-specification}) and from \cref{lemma:probability-of-acceptance}.

    \Cref{corollary:single-oracle-to-oracle-pair} almost gives us the soundness result. To formally complete the proof, we must show that if a quantum polynomial-time counterfeiter $\ccC$ is given a banknote of the form \mbox{$\ket{\$_r} = \ket{z
    _r} U_{\mathcal{B}_r}\ket{0}_{\theta_r}$}, then querying the oracles $\mathcal{G}, \mathcal{Z},\mathcal{T}_{\text{primal}}, \mathcal{T}_{\text{dual}}$ provides no additional advantage over querying $(U_{C_{\mathcal{E}_X}, r}, U_{C_{\mathcal{E}_Z},r})$, as done in \cref{corollary:single-oracle-to-oracle-pair}.

    The goal is to show that any successful attack on $\mathcal{M}$ can be converted into successful counterfeiting of $\ket{C}$ given oracle access to $(U_{C_{\mathcal{E}_X}, r}, U_{C_{\mathcal{E}_Z},r})$. This would then contradict \cref{corollary:single-oracle-to-oracle-pair}, thus proving the result. This can be achieved by observing that an adversary trying to counterfeit $\ket{C}$, given $(U_{C_{\mathcal{E}_X}, r}, U_{C_{\mathcal{E}_Z},r})$, can essentially simulate the oracles $\mathcal{G}, \mathcal{Z}, \mathcal{T}_{\text{primal}}, \mathcal{T}_{\text{dual}}$.

    For the oracle $\mathcal{G}$, let $r \in \{0,1\}^n$ be the random string chosen by the bank, so that $\mathcal{G}(r)=(z_r, \theta_r, \mathcal{B}_r)$. Observe that the string $r$ remains uniformly random, even when we condition on $z_r, \theta_r$, and $\mathcal{B}_r$, as well as complete descriptions of $\mathcal{T}_{\text{primal}}, \mathcal{T}_{\text{dual}}$ and $\mathcal{Z}$. Even if we query $\mathcal{G}(r')$ for $r' \neq r$, no information about $r$ is revealed because the values of $\mathcal{G}$ are generated independently. So, for the simulation of $\mathcal{G}$, suppose we instead modify it by setting $\mathcal{G}(r) \coloneqq (z'_r, \theta'_r, \mathcal{B}'_r)$ for some new $3n$-bit serial number $z'_r$ and $\theta'_r, \mathcal{B}'_r$ (such that $\theta'_r$ has Hamming weight $n/2$) chosen uniformly at random. Then the BBBV hybrid argument \cite{BBBV97} says that, in expectation over $r$, this can alter the final state output by the counterfeiter $\ccC(\$_r)$ by at most $\textsf{poly}(n)/2^{n/2}$ in trace distance. Thus, if $\ccC$ succeeded with non-negligible probability before, then $\ccC$ must still succeed with non-negligible probability after we set $\mathcal{G}(r) \coloneqq (z'_r, \theta'_r, \mathcal{B}'_r)$.

    Given this simulation of $\mathcal{G}$, an adversary can now easily simulate $\mathcal{Z}$ for the new serial number $z'_r$. For $\mathcal{T}_{\text{primal}}$, recall that its behaviour is $\mathcal{T}_{\text{primal}}\ket{z}\ket{v} = \ket{z}U_{C_{\mathcal{E}_X}}\ket{v}$, which can be simulated with knowledge of $z$ and $U_{C_{\mathcal{E}_X}}$. Similarly, $\mathcal{T}_{\text{dual}}$ can be simulated with $U_{C_{\mathcal{E}_Z}}$. Note that the security guarantees we consider are query complexity bounds, thus we do not need to be concerned with the computational complexity of these simulations.

    Thus, with these simulated oracles, any successful attack on $\mathcal{M}$ can indeed be converted into successful counterfeiting of $\ket{C}$ with oracle access to only $(U_{C_{\mathcal{E}_X}, r}, U_{C_{\mathcal{E}_Z},r})$. This contradicts \cref{corollary:single-oracle-to-oracle-pair}.
\end{proof}

\begin{remark}
\label{remark:correcting-noise}
    The scheme we have presented tolerates noise but does not correct it. Consequently, a valid banknote will eventually accumulate too many errors to be accepted. It is straightforward to adapt the subset approach to correct tolerable noise. Instead of checking bit-flip errors with a single oracle $U_{C_{\mathcal{E}_X}}$, we would use a collection of oracles $\{U_{C+e}\}_{e \in \mathcal{E}_X}$ where $U_{C+e}$ checks for membership in the coset $C+e$, and similarly for phase-flip errors. Verification would involve querying these oracles sequentially; by doing this, we can determine if the money state is in one of the cosets (in which case we accept it) and, based off which coset it belongs to, we learn which error has occurred (which allows for correction). The security proof would proceed in a similar manner. In \cref{section:applying-the-method}, the construction of $C^*$ would remain the same, but now a single query to $U_{C + e}$ corresponds to a single query to $U_{C^*}$, and so the query complexity in \cref{corollary:single-oracle-to-oracle-pair} would be $\Omega(\sqrt{\varepsilon}2^{n/4})$. The downside of this approach is that verification now involves the oracles $(\{U_{C+e}\}_{e \in \mathcal{E}_X}, \{U_{C+e}\}_{e \in \mathcal{E}_Z})$, whereas the scheme we presented for tolerating noise only involved $(U_{C_{\mathcal{E}_X}}, U_{C_{\mathcal{E}_Z}})$, and so this scheme for error correction could involve many additional queries as part of the verification algorithm.
\end{remark}

%~~~~~~~~~~~~~~~~~~~%
% Future directions %
%~~~~~~~~~~~~~~~~~~~%
\section{Future directions}
\label{section:future-directions}

\paragraph{In \cref{section:verification}, we presented two approaches to verification: the subset approach and the subspace approach. Is one preferable over the other?} Both approaches achieve the same completeness and soundness error, but is one easier to realize in the plain model? Towards answering this question, we note that in \cite{Zha19}, Zhandry instantiated the oracle scheme of \cite{AC12} by using iO to construct what he called a \emph{subspace hiding obfuscator}. This technique seems to naturally apply to our subspace approach, thus implying that it can also be realized in the plain model. Is it any more difficult to use iO to instantiate the classical oracle that checks membership in subsets?

\paragraph{How does adding a layer of error correction compare to our approach?} In \cref{section:a-layer-of-error-correction}, we briefly discussed how one could take a different approach to noise-tolerance by applying a layer of error correction on top of the \cite{AC12} oracle scheme. This approach has no known security proof, though we have identified the following obstacles: translating any necessary computations to versions that operate on encoded states, and determining if/how the original security proof changes when the adversary holds an encoded state instead of a subspace state. Given a realization of this approach, it would be interesting to make a comparison with our scheme. Notably, a layer of error correction typically requires many additional qubits (\emph{i.e.}, large overhead), so one may ask questions like: for a fixed $n$, how does the noise-tolerance and soundness error compare?

\paragraph{What is the optimal noise-tolerance for a fixed soundness error?} In the discussion that followed \cref{corollary:single-oracle-to-oracle-pair}, we observed that for a fixed soundness error, the noise-tolerance $q/n$ could be improved by increasing $n$. Can this observation be formalized in general? Specifically, for a fixed, arbitrary soundness error, what choice of $n$ maximizes the quantity $q/n$? Answering this exactly requires analysis of the expression for $|\mathcal{E}_q|^2 \varepsilon$. For low noise, the bound in \cref{eq:error-bound} could be used to perhaps simplify the analysis.

\paragraph{Is it possible to obtain noise-tolerance for other public-key quantum money schemes?} In \cref{section:further-related-work}, we mentioned some of the other public-key quantum money schemes that rely on hardness assumptions instead of oracles. Is it possible to make any of these schemes noise-tolerant?

%~~~~~~~~~~~~~~~~~~~%
% Acknowledgments %
%~~~~~~~~~~~~~~~~~~~%
\section*{Acknowledgments}
We would like to thank Anne Broadbent for advice and related discussions. We also thank the anonymous referees for their comments. We acknowledge the support of the Natural Sciences and Engineering Research Council of Canada (NSERC), the US Air Force Office of Scientific Research under award number FA9550-20-1-0375, and the University of Ottawa’s Research Chairs program.

%%%%%%%%%%%%%%%%%%%%%%%%%%%%%%%%%%%%%%%%%%%%%%%%%%%%%%

\bibliographystyle{bibtex/bst/alphaarxiv.bst}
\bibliography{bibtex/bib/quasar-full.bib,
              bibtex/bib/quasar.bib,
              bibtex/bib/quasar-more.bib}

%%%%%%%%%%%%%%%%%%%%%%%%%%%%%%%%%%%%%%%%%%%%%%%%%%%%%%
\end{document}